\documentclass[%
reprint,
superscriptaddress,
nofootinbib,
amsmath,amssymb,
aps,longbibliography
]{revtex4-1}

\newif\ifarxiv 
\arxivtrue

\usepackage{amsthm,bm,xparse,xcolor,tikz}
\usepgfmodule{decorations}
\usetikzlibrary{decorations,decorations.pathmorphing,patterns}

\AtBeginDocument{%
    \newwrite\bibnotes
    \def\bibnotesext{Notes.bib}
    \immediate\openout\bibnotes=\jobname\bibnotesext
    \immediate\write\bibnotes{@CONTROL{REVTEX41Control}}
    \immediate\write\bibnotes{@CONTROL{%
    apsrev41Control,author="08",editor="1",pages="1",title="0",year="0"}}
     \if@filesw
     \immediate\write\@auxout{\string\citation{apsrev41Control}}%
    \fi
  }%

\usepackage[nodayofweek]{datetime}  
\usepackage{hyperref}
\usepackage{xcolor}
\usepackage[capitalise]{cleveref}
\usepackage{enumerate}   
\definecolor{mylinkcolor}{rgb}{0,0,0.8} 
\hypersetup{unicode=true, %
  bookmarksnumbered=false,bookmarksopen=false,bookmarksopenlevel=1, %
  breaklinks=true,pdfborder={0 0 0},colorlinks=true}%
\hypersetup{%
  anchorcolor=mylinkcolor,citecolor=mylinkcolor, %
  filecolor=mylinkcolor,linkcolor=mylinkcolor, %
  menucolor=mylinkcolor,runcolor=mylinkcolor, %
  urlcolor=mylinkcolor}%

\newtheorem{theorem}{Theorem}
\newtheorem{lemma}{Lemma}

\theoremstyle{definition}
\newtheorem{proposition}{Proposition} 
\newtheorem{definition}{Definition}

\newcommand{\ket}[1]{| #1 \rangle}
\newcommand{\bra}[1]{\langle #1 |}
\newcommand{\braket}[2]{\langle #1|#2\rangle}
\newcommand{\ketbra}[2]{|#1\rangle\!\langle#2|}
\newcommand{\id}{\openone}

\begin{document}
\title{Device-independent quantum key distribution with arbitrarily small nonlocality}
\author{Lewis Wooltorton}
    \email{lewis.wooltorton@york.ac.uk}
    \affiliation{Department of Mathematics, University of York, Heslington, York, YO10 5DD, United Kingdom}
    \affiliation{Quantum Engineering Centre for Doctoral Training, H. H. Wills Physics Laboratory and Department of Electrical \& Electronic Engineering, University of Bristol, Bristol BS8 1FD, United Kingdom}
\author{Peter Brown}
    \email{peter.brown@telecom-paris.fr}
    \affiliation{Télécom Paris - LTCI, Inria, Institut Polytechnique de Paris,
19 Place Marguerite Perey, 91120 Palaiseau, France}
\author{Roger Colbeck}
    \email{roger.colbeck@york.ac.uk}
    \affiliation{Department of Mathematics, University of York, Heslington, York, YO10 5DD, United Kingdom}

\date{$24^{\text{th}}$ May 2024}

\begin{abstract}
Device-independent quantum key distribution (DIQKD) allows two users to set up shared cryptographic key without the need to trust the quantum devices used. Doing so requires nonlocal correlations between the users. However, in [Phys.\ Rev.\ Lett.\ {\bf 127}, 050503 (2021)] it was shown that for known protocols nonlocality is not always sufficient, leading to the question of whether there is a fundamental lower bound on the minimum amount of nonlocality needed for any DIQKD implementation. Here we show that no such bound exists, giving schemes that achieve key with correlations arbitrarily close to the local set. Furthermore, some of our constructions achieve the maximum of 1 bit of key per pair of entangled qubits. We achieve this by studying a family of Bell inequalities that constitute all self-tests of the maximally entangled state with a single linear Bell expression. Within this family there exist non-local correlations with the property that one pair of inputs yield outputs arbitrarily close to perfect key. Such correlations exist for a range of Clauser-Horne-Shimony-Holt (CHSH) values, including those arbitrarily close to the classical bound. Finally, we show the existence of quantum correlations that can generate both perfect key and perfect randomness simultaneously, whilst also displaying arbitrarily small CHSH violation. This opens up the possibility of a new class of cryptographic protocol.
\end{abstract}
\maketitle

\ifarxiv\section{Introduction}\label{sec:intro}\else\noindent{\it Introduction.|}\fi
Establishing shared or global randomness between two isolated parties is a task achievable using quantum theory~\cite{BB84,Ekert,Scarani09}, but inaccessible to classical physics without additional assumptions. Quantum key distribution (QKD), for example, can be performed by making measurements on a shared entangled state, and security is derived assuming the devices behave according to a physical model~\cite{pirandola2020advances}. Meeting the practical requirements of a model can be challenging, and mismatches between the model and reality can lead to security problems (see e.g., \cite{GLLSKM}). However, quantum theory allows us to bypass the majority of such mismatch issues: entangled quantum systems can exhibit input-output behaviours that are nonlocal~\cite{Bell_book,BarrettNonlocalResource}, giving rise to device-independent approaches to QKD~\cite{Ekert,BHK,ABGMPS,PABGMS,bcktwo,VV2,ADFRV}. The same can be said about the related task of randomness expansion~\cite{ColbeckThesis,PAMBMMOHLMM,CK2,MS1,MS2,CR_free}.

Given quantum correlations that exhibit \textit{some} nonlocality, how much secure key can be extracted device-independently? To achieve the highest security, we want to find a lower bound on the amount of key conditioned on the observed nonlocality. Such lower bounds have been found in a variety of scenarios, both analytically~\cite{PABGMS,ho2020noisy,Woodhead2021deviceindependent} and numerically~\cite{BrownDeviceIndependent,TanDI,MasiniDI,BrownDeviceIndependent2}.

A related question that has achieved less attention is: for what range of nonlocality is DIQKD possible? It has recently been shown that nonlocality is not a sufficient condition for DIQKD using standard protocols~\cite{Farkas21}. More precisely, \cite{Farkas21} showed that there exist quantum correlations, arising from Werner states~\cite{Werner}, with some nonlocality, for which an upper bound on the secret key rate vanishes. Whilst the result of~\cite{Farkas21} does not encompass all possible protocols, it raises the question of whether there exists a minimum amount of nonlocality needed for any DIQKD implementation. A conclusive proof of existence, or contradiction, for such a bound is currently missing from the literature. 

In this paper, we show that no such bound exists. Contrasting the work of~\cite{Farkas21}, we show one can find quantum correlations with arbitrarily small nonlocality that can be used for DIQKD with a key rate arbitrarily close to 1 bit per pair of shared entangled qubits (we refer to this as near-perfect DIQKD). This complements existing work showing the same holds for global randomness expansion (DIRE)~\cite{AcinRandomnessNonlocality,WBC}, which we expand upon here. We also go one step further: there exist quantum correlations with arbitrarily small nonlocality that can be used for both near-perfect DIQKD and maximum DIRE. To our knowledge this is the first example of such correlations to appear in the literature, and could open up the possibility for a new class of cryptographic protocols. 

Our results are obtained by self-testing~\cite{mayers2004self,McKagueSinglet,YangSelfTest,KaniewskiSelfTest,SupicSelfTest} quantum correlations close to the local boundary. We study a versatile family of bipartite Bell expressions that first appeared in~\cite{Le2023quantumcorrelations}, and encompass those used in the literature to certify secret key~\cite{Woodhead2021deviceindependent} and randomness~\cite{WBC,WBC2}. These expressions are tangent hyperplanes to the boundary of the set of quantum correlators, and constitute all self-tests of the singlet with a single linear Bell expression, when considering two observers with binary inputs and outputs~\cite{Le2023quantumcorrelations,barizien2023}. Moreover, to prove self-testing we reduce the problem to qubits via Jordan's lemma~\cite{Jordan}; a self-contained reduction can be found in the \ifarxiv Appendices, \else Supplemental Material~\cite{supp}, \fi which may be of independent interest.

\ifarxiv\section{Background}\label{sec:background}\else\medskip\noindent{\it Background.|}\fi
We consider the minimal Bell scenario for DIQKD. Let two space-like separated parties, Alice and Bob, each hold a device with inputs $x,y \in \{0,1\}$ and outputs $a,b \in \{0,1\}$. The devices are characterized by the joint distribution $p(ab|xy)$, which must be no signalling.

A quantum strategy refers to a joint state, $\rho_{\tilde{Q}_{A}\tilde{Q}_{B}}$, and sets of observables $\tilde{A}_{x} = \tilde{M}_{0|x} - \tilde{M}_{1|x}$, $\tilde{B}_{y} = \tilde{N}_{0|y} - \tilde{N}_{1|y}$, where $\{\tilde{M}_{a|x}\}_{a}$, $\{ \tilde{N}_{b|y}\}_{b}$ are projective measurements (which can be assumed without loss of generality according to Naimark's dilation theorem~\cite{paulsen_2003}) on the physical Hilbert spaces $\mathcal{H}_{\tilde{Q}_{A}}$ or $\mathcal{H}_{\tilde{Q}_{B}}$ held by Alice and Bob. We consider an adversarial scenario in which Eve holds a purification $\ket{\Psi}_{\tilde{Q}_{A}\tilde{Q}_{B}E}$ of $\rho_{\tilde{Q}_{A}\tilde{Q}_{B}}$ and can set the quantum behaviour of each device (i.e., which measurements each input corresponds to). Eve's aim is to establish nontrivial correlations between the classical register $A$ holding Alice's outcomes and $E$, while remaining undetected. Such correlations will allow Eve to learn information about Alice's raw key when Alice measures e.g., $X=x$; this is described by the post-measurement classical-quantum state $\rho_{AE|X=x} = \sum_{a} \ketbra{a}{a}_{A} \otimes \rho_{E}^{a|x}$, where $\rho_{E}^{a|x}=\mathrm{Tr}_{\tilde{Q}_{A}\tilde{Q}_{B}}[(\tilde{M}_{a|x}\otimes \mathbb{I}_{\tilde{Q}_{B}E})\ketbra{\Psi}{\Psi}]$ is the subnormalized state held by Eve conditioned on Alice getting $a$ when $x$ is measured. The global post-measurement classical-classical-quantum state $\rho_{ABE|X=x,Y=y}$ is defined analogously, and the behaviour is recovered via the Born rule $p(ab|xy) = \mathrm{Tr}_{\tilde{Q}_{A}\tilde{Q}_{B}E}[(\tilde{M}_{a|x}\otimes \tilde{N}_{b|y} \otimes \mathbb{I}_{E})\ketbra{\Psi}{\Psi}]$.

As we are restricting ourselves to binary inputs and outputs, the nonlocality of the resulting joint behaviour can be quantified in terms of its CHSH value\footnote{There are eight CHSH-type inequalities, equivalent to $I_{\mathrm{CHSH}}$ up to relabellings. These are the only facet inequalities separating classical and quantum correlations in this scenario.}, $I_{\mathrm{CHSH}} = \langle A_{0}(B_{0} + B_{1})\rangle + \langle A_{1}(B_{0} - B_{1})\rangle $, where $\langle A_{x}B_{y} \rangle$ are the correlators, $\langle A_{x}B_{y} \rangle = \sum_{ab} (-1)^{a+b}p(ab|xy)$, which equal $\mathrm{Tr}[(\tilde{A}_{x} \otimes \tilde{B}_{y})\rho_{\tilde{Q}_{A}\tilde{Q}_{B}}]$ when the behaviour is quantum. The local and quantum bounds are given by $2$ and $2\sqrt{2}$ respectively, and it is well known that there is a unique quantum state and sets of measurements that achieve the quantum bound, up to local isometries. It is in this sense that the CHSH inequality \textit{self-tests} the corresponding state (which is maximally entangled) and measurements.

We consider DIQKD protocols based on spot-checking with two measurements per party and a single Bell inequality (see, e.g.,~\cite[Section 4.4]{pirandola2020advances} for an example using the CHSH inequality\footnote{Other than the choice of Bell inequality, a key difference in the present work is that the settings used for raw key are among those used for testing; in the protocol of~\cite{pirandola2020advances} Bob needs an additional setting. We elaborate further on this in the \ifarxiv Appendices\else Supplemental Material~\cite{supp}\fi.}). In order to compute the secret key rate the relevant quantities are the conditional von Neumann entropies $H(A|X=x,E)$ and $H(A|X=x,Y=y,B)$, where $X=x$ and $Y=y$ are the inputs used for key generation. The latter entropy, which is independent of Eve's system, captures the cost for Alice and Bob to reconcile their raw keys and can be estimated directly from the statistics. The former captures the randomness in Alice's raw key conditioned on Eve, and must be lower bounded in terms of the observed behaviour $P_{\mathrm{obs}}$, or some functions $f_{i}$ of $P_{\mathrm{obs}}$ (for instance, $f_i$ might be a Bell expression). The asymptotic secret key rate is then bounded by the Devetak-Winter formula~\cite{DW}
\begin{multline}
    r^{\mathrm{key}} \geq \max_{x,y} \Big( \inf \big[ H(A|X=x,E)_{\rho_{AE|X=x}} \big] \\ - H(A|X=x,Y=y,B)_{\rho_{AB|X=x,Y=y}}\Big), \label{eq:rate}
\end{multline}
where the infimum is taken over states and measurements compatible with $f_{i}(P_{\mathrm{obs}})$. Analogously, the global randomness rate is defined by the quantity $r^{\mathrm{global}} = \max_{x,y}\inf H(AB|X=x,Y=y,E)_{\rho_{ABE|X=x,Y=y}}$. The asymptotic rates can serve as a basis for rates with finite statistics using tools such as the entropy accumulation theorem~\cite{DFR, ADFRV,EAT2}.

To achieve a key rate arbitrarily close to $1$ bit per entangled state shared, we consider the family of three-parameter Bell inequalities from~\cite{Le2023quantumcorrelations} whose maximal quantum violation self-tests a unique state and measurements (up to local isometries). In this case we consider a single functional $f = \langle B_{\theta,\phi,\omega} \rangle$ with observed value $\eta$, and denote the rate $R^{\mathrm{key}}_{\theta,\phi,\omega}(\eta)$. Achieving the quantum bound self-tests a pure (in fact, maximally entangled) state, which therefore must be uncorrelated with Eve, allowing us to directly compute the entropy from the observed behaviour. We find $H(A|X=0,E)=1$, and $H(A|X=Y=0,B)=\epsilon$ for any epsilon $\epsilon\in(0,2-(3/4)\log(3)]$, giving a key rate $1-\epsilon$ that tends to 1 as $\epsilon\to0$, while at the same time having a CHSH value arbitrarily close to classical bound. We also use $R^{\mathrm{global}}_{\theta,\phi,\omega}(\eta)$ to denote the randomness rate based on the same functional.

\ifarxiv\section{Methods}\label{sec:methods}\else\medskip\noindent{\it Methods.|}\fi
Our main results are derived from studying the family of self-testing Bell expressions in~\cite{Le2023quantumcorrelations}, also recently reported in~\cite{barizien2023}. We provide a new self-testing proof, and all claims are proven in \ifarxiv Appendices~\ref{app:selfProof} and~\ref{app:3to6}\else the Supplemental Material~\cite{supp}\fi.  
\begin{proposition} \label{prop:bellExp}
Let $\theta,\phi,\omega \in \mathbb{R}$.
Define the family of Bell expressions, labelled $\langle B_{\theta,\phi,\omega} \rangle$,
\begin{multline}\label{eq:fam}
    \cos(\theta + \phi)\cos(\theta+\omega)\, \big \langle A_{0}\big( \cos \omega \, B_{0} - \cos \phi \, B_{1} \big) \big \rangle + \\ \cos \phi \cos \omega \big \langle  A_{1}\big( -\cos(\theta + \omega) \, B_{0} + \cos (\theta + \phi) \, B_{1} \big) \big \rangle. 
\end{multline}
Then the following hold: (i) The local bounds are given by $\pm \eta^{\mathrm{L}}_{\theta,\phi,\omega}$, where $\eta^{\mathrm{L}}_{\theta,\phi,\omega}$ is defined as
\begin{multline}
    \max_{\pm} \Big\{\left| \cos(\theta+\omega)\cos(\omega)\big( \cos(\theta+\phi)\pm\cos(\phi) \big)\right|+\\ \left| \cos(\theta+\phi)\cos(\phi)\big( \cos(\theta+\omega)\pm\cos(\omega) \big)\right| \Big\}.
\end{multline}
(ii)  If
\begin{equation}\label{cond}
\cos(\theta+\phi)\cos(\phi)\cos(\theta+\omega)\cos(\omega)<0,
\end{equation}
then the quantum bounds are given by $\pm \eta^{\mathrm{Q}}_{\theta,\phi,\omega}$, where \begin{equation}
\eta^{\mathrm{Q}}_{\theta,\phi,\omega} = \sin(\theta)\sin(\omega-\phi)\sin(\theta+\omega+\phi)\,.
\end{equation}

\vspace{0.2cm}

\noindent (iii) $\left|\eta^{\mathrm{Q}}_{\theta,\phi,\omega}\right|> \eta^{\mathrm{L}}_{\theta,\phi,\omega}\iff$ \eqref{cond} holds.

\vspace{0.2cm}

\noindent (iv)
If~\eqref{cond} holds then up to local isometries there is a unique strategy that achieves $\langle B_{\theta,\phi,\omega} \rangle = \eta_{\theta,\phi,\omega}^{\mathrm{Q}}$: 
    \begin{gather}
        \rho_{Q_{A}Q_{B}} = \ketbra{\psi}{\psi}, \ \mathrm{where} \ \ket{\psi} = \frac{\ket{00} + i\ket{11}}{\sqrt{2}}, \nonumber \\
        A_{0} = \sigma_{X}, \
        A_{1} = \cos \theta \, \sigma_{X} + \sin \theta \, \sigma_{Y}, \label{eq:strat} \\
        B_{0} = \cos \phi \, \sigma_{X} + \sin \phi \, \sigma_{Y}, \
        B_{1} = \cos \omega \, \sigma_{X} + \sin \omega \, \sigma_{Y}. \nonumber
    \end{gather}

\end{proposition}
\noindent Special cases of the above family have already found applications in DI cryptography. For example, it contains all bipartite expressions found in~\cite{WBC,WBC2}, which certify maximum global DI randomness. One can also recover the marginal-free subfamily of the tilted CHSH inequalities~\cite{AcinRandomnessNonlocality,BampsPironio} which have found use in DIQKD~\cite{Woodhead2021deviceindependent,Sekatski2021} and robust self-testing of the singlet~\cite{ValcarceSelfTest}. Moreover, it has been shown~\cite{Le2023quantumcorrelations,barizien2023} that the family~\eqref{eq:fam} constitute an infinite family of hyperplanes tangent to the boundary of the quantum set of correlations with uniform marginals, $\mathcal{Q}_{\mathrm{corr}}$, and constitute every self-test of the singlet with a single linear function of the correlators in this scenario. We discuss these connections in more detail in \ifarxiv Appendices~\ref{app:C}--\ref{app:E}\else the Supplemental Material~\cite{supp}\fi. 

\ifarxiv\section{Perfect randomness from arbitrarily small nonlocality}\label{sec:res1}\else\medskip\noindent{\it Perfect randomness from arbitrarily small nonlocality.|}\fi First we consider all strategies that certify maximum global randomness in this scenario using a maximally entangled state, certified by a single Bell expression. This contains the subfamily in~\cite{WBC}, where the randomness versus nonlocality relationship was studied using a one parameter family of Bell expressions; here we review and extend this to a two parameter sub-family of~\eqref{eq:fam}.

\begin{proposition}[\cite{WBC}]
    For any $s\in(2,3\sqrt{3}/2]$, there exists a tuple $(\theta,\phi,\omega)$ satisfying~\eqref{cond}, along with a set of quantum correlations achieving $R^{\mathrm{global}}_{\theta,\phi,\omega}(\eta_{\theta,\phi,\omega}^{\mathrm{Q}}) = 2$ and $I_{\mathrm{CHSH}}=s$. \label{prop:rand}
\end{proposition}
The subfamily self-tests $\sigma_{X}$ for both Alice and Bob, along with $\ket{\psi} = (\ket{+,y_{+}} + \ket{-,y_{-}})/\sqrt{2}$, where $\ket{\pm} \ (\ket{y_{\pm}})$ are the eigenstates of $\sigma_X$ ($\sigma_Y$). This is achieved by setting $\omega = \pi$, resulting in a two parameter family of Bell expressions certifying uniform randomness when $X=0,Y=1$ and~\eqref{cond} holds: $\cos(\theta + \phi)\cos(\theta)\, \big \langle A_{0}\big( B_{0} + \cos (\phi) \, B_{1} \big) \big \rangle - \cos (\phi) \big \langle  A_{1}\big( \cos(\theta) \, B_{0} + \cos (\theta + \phi) \, B_{1} \big) \big \rangle$. \cref{fig:rand} shows the valid regions of $(\theta,\phi)$ space for maximum randomness certification.

\begin{figure}[h]
\includegraphics[width=8.4cm]{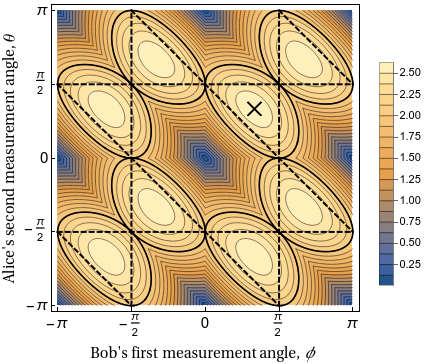}
\centering
\caption{Contour plot of nonlocality, measured using the maximum of the eight CHSH-type inequalities, for the strategies in \cref{eq:strat} with $\omega = \pi$. The points inside the dashed triangles, excluding the boundary, can be used for perfect DIRE with a single linear Bell inequality: they satisfy~\eqref{cond} and have a value in $(2,3\sqrt{3}/2]$ for one of the CHSH-type inequalities, with the maximum of $I_{\mathrm{CHSH}}$ indicated with the black cross at $\theta = \phi = \pi/3$. Approaching $\phi=-\pi/2$ or $\phi=\pi/2$ inside the corresponding region also allows arbitrarily good DIQKD. The black contours indicate $I_{\mathrm{CHSH}} = 2$ for at least one CHSH-type inequality.}
\label{fig:rand}
\end{figure}

\ifarxiv\section{Near perfect key from arbitrarily small nonlocality}\label{sec:res2}\else\medskip\noindent{\it Near perfect key from arbitrarily small nonlocality.|}\fi Next we turn our attention to DIQKD. We consider the key rate achievable using the strategies in \cref{prop:bellExp}, and which CHSH values are compatible. 

\begin{proposition}
    For any $s\in(2,5/2]$, and any $\epsilon\in(0,2-(3/4)\log(3)]$, there exists a tuple $(\theta,\phi,\omega)$ satisfying~\eqref{cond}, along with a set of quantum correlations achieving $R^{\mathrm{key}}_{\theta,\phi,\omega}(\eta_{\theta,\phi,\omega}^{\mathrm{Q}}) = 1- \epsilon$ and $I_{\mathrm{CHSH}}=s$. \label{prop:key}
\end{proposition}
This statement is achieved by self-testing $\sigma_{X}$ for Alice and $\cos(\phi)\,\sigma_{X} + \sin(\phi)\,\sigma_{Y}$ for Bob, along with $\ket{\psi}$. One can then take $\phi$ arbitrarily close to $\pi/2$; we find $H(A|X=0,E) = 1$ from self-testing, and $H(A|X=0,Y=0,B) = H_{\mathrm{bin}}[(1+\sin(\phi))/2]:= \epsilon$, where $H_{\mathrm{bin}}$ is the binary entropy function. Hence in the limit $\phi \rightarrow \pi/2$, $\epsilon$ tends to $0$ and we achieve perfect key. Moreover, at this limit point, we can choose $(\theta,\omega)$ such that the CHSH interval $(2,5/2]$ is achieved---see \cref{fig:key} for an illustration.

Interestingly, the limit point violates~\eqref{cond}: at $\phi = \pi/2$ the Bell expression~\eqref{eq:fam} becomes trivial (the local and quantum bounds coincide), and we cannot find a single Bell expression that can certify perfect key. The correlations achieved in this limit (those resulting from the construction~\eqref{eq:strat}) are non-local and were recently studied in~\cite{barizien2023}, where it was shown such points correspond to non-exposed regions of $\mathcal{Q}_{\mathrm{corr}}$; our result shows the implications of this for DIQKD. At least two linear functionals are then required to uniquely identify the correlations, and, indeed, using all four correlators one can verify for various values of $(\theta,\omega)$, the limit point satisfies the self-testing criteria of the singlet given by Wang \textit{et al.}~\cite{Wang_2016}, and a one-parameter subfamily containing the point with $I_{\mathrm{CHSH}}=5/2$ was studied in~\cite[Section~3.4.1]{Chen2023}, including the non-exposed nature of the case with $I_{\mathrm{CHSH}}=5/2$. We can therefore recover another statement similar to \cref{prop:key}. We express this using $r^{\mathrm{key}}_{\theta,\phi,\omega}$, which is the quantity defined by \cref{eq:rate} evaluated with functionals $f_{i}$ constraining all four correlators to the values generated by the strategy in \cref{eq:strat}.
\begin{proposition}
    For any $s\in(2,5/2]$, there exists a tuple $(\theta,\phi,\omega)$, along with a set of quantum correlations achieving $r^{\mathrm{key}}_{\theta,\phi,\omega} = 1$ and $I_{\mathrm{CHSH}}=s$. \label{prop:key_corr}
\end{proposition}
In other words, if we constrain all four correlators rather than the Bell inequality of~\eqref{eq:fam}, then we achieve perfect key directly, rather than in a limit. However, the use of four correlators can have disadvantages in the case of finite statistics because for a fixed number of shared states larger error bars are present when estimating four quantities rather than one, see e.g.,~\cite{BRC,DFR}. 

\begin{figure}[h]
\includegraphics[width=8.4cm]{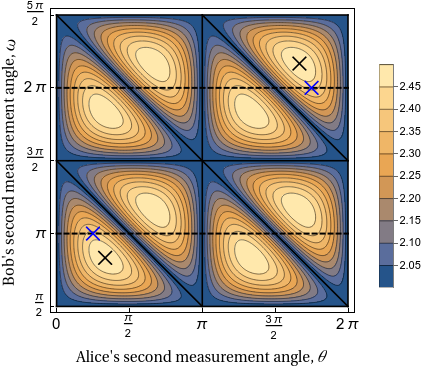}
\centering
\caption{Contour plot of nonlocality, measured using the maximum of the eight CHSH-type inequalities of the strategies in \cref{eq:strat}, at the limit $\phi=\pi/2$.  All points on the graph are limit points of correlations that achieve arbitrarily perfect DIQKD with a single linear Bell inequality, including the contours with CHSH values equal to 2, which are shown as black triangles. The black dashed lines show where perfect DIRE can also be achieved, with the blue crosses denoting the maximum value of $I_{\mathrm{CHSH}}=1+\sqrt{2}$ at $\theta=\pi/4,\omega=\pi$ and $\theta = 7\pi/4, \omega = 2\pi$. The black crosses denote the global maximum of $I_{\mathrm{CHSH}}=5/2$ at $\theta = \pi/3, \omega = 5\pi/6$ and $\theta = 5\pi/3,\omega = 13\pi/6$.}
\label{fig:key}
\end{figure}

\ifarxiv\section{Near perfect key and randomness from arbitrarily small nonlocality}\label{sec:res3}\else\medskip\noindent{\it Near perfect key and randomness from arbitrarily small nonlocality.|}\fi Finally, we consider the possibility of using the same set of quantum correlations to generate perfect key from one input combination, and perfect randomness from another. 
\begin{proposition}
    For any $s\in(2,1+\sqrt{2}]$ and any $\epsilon\in\left(0,H_{\mathrm{bin}}((2+\sqrt{2})/4)\right)$, there exists a tuple $(\theta,\phi,\omega)$ satisfying~\eqref{cond}, along with a set of quantum correlations achieving $R^{\mathrm{global}}_{\theta,\phi,\omega}(\eta_{\theta,\phi,\omega}^{\mathrm{Q}})=2$, $R^{\mathrm{key}}_{\theta,\phi,\omega}(\eta_{\theta,\phi,\omega}^{\mathrm{Q}})=1-\epsilon$ and $I_{\mathrm{CHSH}}=s$.
\end{proposition}

This is obtained by simultaneously self-testing $\sigma_{X}$ for Alice, and both $\cos(\phi)\,\sigma_{X} + \sin(\phi)\,\sigma_{Y}$ and $\sigma_{X}$ for Bob, along with $\ket{\psi}$. Following the same arguments as before, we can fix $\omega = \pi$, and take $\phi$ arbitrarily close to $\pi/2$, resulting in a key rate of $1 - \epsilon$ and global randomness $2$ for input choices $X=0,Y=0$ and $X=0,Y=1$ respectively. By varying $\theta$, we can achieve the range of CHSH values $(2,1+\sqrt{2}]$, as shown in \cref{fig:key}. For the same reasons discussed in the previous section, there also exists a non-limiting statement when the full correlators are considered.
\begin{proposition}
    For any $s \in (2,1+\sqrt{2}]$, there exists a tuple $(\theta,\phi,\omega)$, along with a set of quantum correlations achieving $r^{\mathrm{global}}_{\theta,\phi,\omega} = 2$, $r^{\mathrm{key}}_{\theta,\phi,\omega} = 1$, and $I_{\mathrm{CHSH}} = s$.
\end{proposition}

\ifarxiv\section{Discussion}\label{sec:disc}\else\medskip\noindent{\it Discussion.|}\fi We have shown that quantum theory allows perfect DI key to be shared between two users using correlations that are arbitrarily close to being local. However, we do not know that any correlation exhibiting nonlocality, can be used for DIQKD, for instance, as shown in~\cite{Farkas21}, those that lie in the interior of $\mathcal{Q}_{\mathrm{corr}}$ and arise from measuring experimentally relevant states cannot generate key using standard protocols. The behaviours we use for our results are generated by the singlet, and lie on the self-testable boundary of $\mathcal{Q}_{\mathrm{corr}}$. 

Similar statements hold for DI randomness generation, and we also showed the existence of quantum correlations that can simultaneously be used either to share key or generate maximum randomness, while being arbitrarily close to the local set. This is not only an intriguing feature of quantum theory, but opens up the possibility for new protocols exploiting this feature. For example, certifying global randomness implies 1 bit of blind randomness for Alice, in which she does not need to trust Bob~\cite{MillerBlind,Fu18,MetgerGEAT}. This prompts an application to QKD post-processing in which certified randomness from some outcomes could help replenish some of the private randomness consumed in others. We leave the study of such protocols, and other applications of this construction, to future investigation. 

Although we achieve arbitrarily good key using a single Bell expression, getting perfect key is excluded. On the other hand, perfect key is possible by testing all four correlators. This raises the question of the minimum number of linear quantities required. It would be of further interest to find the robustness of the present constructions to noise. We leave these as problems for future investigation.

Finally, our work also highlights how, when given access to the full set of single Bell functionals that self-test the singlet in this scenario, one can find interesting relationships between cryptographic tasks and nonlocality. It would be interesting to find further applications. For example, it has been shown how use of the various subfamilies of \cref{prop:bellExp} can boost practical DIQKD, DIRE and robust self-testing~\cite{Woodhead2021deviceindependent,WBC,Yang2014,Bancal2015}; given access to their generalizations, further improvements may be found by optimizing over the entire family of Bell expressions~\eqref{eq:fam}.

\medskip{\noindent}{\it Note added:} During the writing up of this work we became aware of a related work~\cite{Farkas23} that also shows the possibility of key distribution with arbitrarily small nonlocality using an alternative approach.

\ifarxiv\acknowledgements\else\medskip\noindent{\it Acknowledgements}|\fi This work was supported by EPSRC via the Quantum Communications Hub (Grant No.\ EP/T001011/1) and Grant No.\ EP/SO23607/1 and by the European Union's Horizon Europe research and innovation programme under the project ``Quantum Secure Networks Partnership'' (QSNP, grant agreement No.\ 101114043).



%

\onecolumngrid 

\appendix

\section{Spot checking protocols for DIQKD} \label{app:spot}

In this section, we give a brief outline of a standard spot checking DIQKD protocol using the Bell inequalities in this work (those of Proposition~1), and compare to protocols in which Bob has an additional measurement setting.

The set up is such that Alice and Bob each have an untrusted device (these are intended to be used for quantum measurements), as well as a trusted way to process classical information, and their own trusted source of private randomness.  They are connected to one another with an insecure but authenticated channel (so an adversary can read any messages sent on this channel, but cannot alter them without being detected due to the authentication).

In each round, Alice uses a private random number generator to assign that round as either a test round or generation round, where typically the probability of a test round is low. She communicates this to Bob, with each party ensuring that their untrusted device does not learn whether the round is a test or generation round. On test rounds, Alice and Bob each sample an input $X=x\in\{0,1\}$, and $Y=y\in\{0,1\}$, according to a uniform distribution. On generation rounds, Alice and Bob deterministically set $X=0$ and $Y=0$. They record their outputs $A=a\in\{0,1\}$, $B=b\in\{0,1\}$.

After many rounds have taken place, Alice sends Bob the settings and outcomes of her test rounds and Bob uses the joint statistics to estimate the Bell value $\langle B_{\theta,\phi,\omega} \rangle$ for some choice of $(\theta,\phi,\omega)$ satisfying the conditions required for Proposition~3. If the value is too low, they abort the protocol. Otherwise the outputs from the generation rounds form the raw key. Alice and Bob proceed with error correction and privacy amplification, distilling a secret key. Note that when using one of the Bell inequalities in this work for which we get arbitrarily close to perfect key (i.e., where $\phi$ is arbitrarily close to $\pi/2$), the amount of information leaked during error correction and the amount of compression required in privacy amplification become negligible as the number of rounds become asymptotically large.

In the above protocol, Bob's key generation setting $Y=0$ is also used in the Bell test. An alternative approach is for Bob to use an extra setting $Y=2$ for key generation, when, e.g., protocols based on the CHSH inequality are used (see Section~4.4 of~\cite{pirandola2020advances} for an example). The reason for this is that the optimal strategy for violating the CHSH inequality yields relatively poor correlation in the case $X=0$, $Y=0$, while Bob's extra measurement for $Y=2$ can be chosen to yield perfect correlation with Alice's $X=0$ measurement in the ideal case, leading to higher key rates. However, in this setting the protocol must include an additional test that checks the correlation when $X=0$ and $Y=2$. Because we use a Bell inequality that is tailored for key generation, the optimal measurements for Alice and Bob already include a pair that give arbitrarily good correlation, and there is no need to add a third measurement.

It might also be possible to certify perfect key with arbitrarily small nonlocality using existing constructions tailored for randomness~\cite{AcinRandomnessNonlocality,WBC}, plus a third measurement for Bob. Instead, our protocol remains in the minimal Bell scenario. By doing so, because the CHSH inequality is the the only non-trivial facet of the local polytope up to symmetry, the CHSH value gives an immediate measure of distance from the local set, and hence can be used to measure the nonlocality. This simplifies the analysis, and gives a protocol that requires fewer steps (see below). We also remark that, whilst not discussed by the authors, one of the constructions of~\cite{AcinRandomnessNonlocality} can be used to certify $\epsilon$-perfect key as the CHSH violation tends to zero with two inputs per party (see equations~(7) and~(13) of~\cite{AcinRandomnessNonlocality}, with $\beta=0$, $\alpha\rightarrow\infty$, $\theta=\pi/4$ and $\varphi=0$). These constructions can be recovered as a subfamily of ours, and we improve upon this understanding by showing $\epsilon$-perfect key is compatible with CHSH values in the whole range $(2,5/2]$, rather than only in the limit as the CHSH value tends to $2$, which is not robust.

On a more practical note, an extra measurement for Bob implies the need to release raw key during the classical communication phase; this data is used to test alignment and estimate the amount of error correction needed. In our protocol, because the settings $X=0$ and $Y=0$ (those used in the generation rounds) also appear in the Bell test, no additional raw key needs releasing for an alignment test. Though a negligible improvement in the asymptotic regime, this may lead to an overall efficiency boost for the protocol in the finite setting. However, we expect noise robustness will play a significant role when testing non-facet defining Bell inequalities, and a finite key rate analysis is needed for a fair comparison, which is beyond the scope of this work.

\section{Self-testing using Jordan's lemma} \label{sec:jlem}

In this section, we show how self-testing statements can be obtained in the bipartite scenario with two inputs and outputs per party using a reduction to qubit systems via Jordan's lemma. We begin by precisely defining self-testing, and give some background on sum-of-squares (SOS) decompositions. We consider only the exact self-testing statement, and leave proof of the robust statement for future work. Then, we apply Jordan's lemma to show a quantum strategy is unique if there is a unique two-qubit strategy. Towards finding this unique strategy, we provide further reductions for Bell inequalities without marginal terms.

\subsection{Definition of self-testing and sum-of-squares decompositions}

\begin{definition}[Self-test]\label{def:selfTest}
Let the observables $\{A_x\}_x$, $\{B_y\}_y$ and pure state $\ket{\psi}_{Q_{A}Q_{B}}$ be the target strategy, and let $B$ be a Bell operator. The inequality $\langle B \rangle \leq \eta^{\mathrm{Q}}$ \emph{self-tests} the target state and measurements if for all physical quantum strategies $(\tilde{\rho}_{\tilde{Q}_{A}\tilde{Q}_{B}},\{\tilde{A}_{x}\}_x,\{\tilde{B}_{y}\}_y)$ that satisfy $\langle B \rangle = \eta^{\mathrm{Q}}$, there exist local isometries $V_A$ and $V_B$ and an ancillary state $\ket{\xi}_{\text{Junk}}$ such that defining $V:\mathcal{H}_{E}\otimes\mathcal{H}_{\tilde{Q}_{A}}\otimes \mathcal{H}_{\tilde{Q}_{B}}\rightarrow\mathcal{H}_{Q_{A}}\otimes\mathcal{H}_{Q_{B}}\otimes\mathcal{H}_{\mathrm{Junk}}$, $V=\mathbb{I}_{E}\otimes V_{A}\otimes V_{B}$ and letting $\ket{\Psi}_{E\tilde{Q}_{A}\tilde{Q}_{B}}$ be a purification of $\tilde{\rho}_{\tilde{Q}_{A}\tilde{Q}_{B}}$, we have
\begin{equation}
    V\Big[ (\mathbb{I}_{E}\otimes\tilde{A}_{x}\otimes\tilde{B}_{y})\ket{\Psi}_{E\tilde{Q}_{A}\tilde{Q}_{B}} \Big] 
    =   \ket{\xi}_{\text{Junk}}\otimes(A_{x} \otimes B_{y}) \ket{\psi}_{Q_{A}Q_{B}}\ \forall\ x,y. 
\end{equation}
\end{definition}
\noindent Throughout the \ifarxiv appendices\else Supplemental Material\fi, we refer to the physical state and measurements we are trying to self-test as the ``reference'', denoted with a tilde. The constituents of the intended strategy, detailed in \cref{eq:appstrat}, are then referred to as the ``target'' state and measurements, for which the relevant Hilbert space is $\mathcal{H}_{Q_{A}}\otimes \mathcal{H}_{Q_{B}}$.  In this work, there will only be two possible values of $x$ and of $y$, and each of $\mathcal{H}_{Q_{A}}$ and $\mathcal{H}_{Q_{B}}$ will be two dimensional.

For a Bell operator $B$ that defines the quantum Bell inequality $\langle B \rangle \leq \eta^{\mathrm{Q}}$, the operator $\bar{B}:= \eta^{\mathrm{Q}}\mathbb{I} - B$, satisfies $\bra{\phi}\bar{B}\ket{\phi}\geq 0$ for all quantum states $\ket{\phi}$, i.e., $\bar{B} \succeq 0$. If there exists a set of operators $\{P_{\mu}\}$ that are polynomials of $A_{x}$ and $B_{y}$ and satisfy 
\begin{equation}
    \bar{B} = \sum_{\mu} P_{\mu}^{\dagger}P_{\mu}, \label{eq:BellSOS}
\end{equation}
then we have found a sum-of-squares (SOS) decomposition of the operator $\bar{B}$: positivity of $\bar{B}$ follows directly from the fact that $K^{\dagger}K\succeq 0$ for any operator $K$. 

SOS decompositions can be used to enforce algebraic constraints on any state and measurements that saturate the quantum Bell inequality, since $\langle B \rangle = \eta^{\mathrm{Q}}$ implies
\begin{equation}
    \langle \bar{B} \rangle  = \sum_{\mu} \bra{\psi}P_{\mu}^{\dagger}P_{\mu}\ket{\psi} = 0,
\end{equation}
which can only hold if $P_{\mu}\ket{\psi} = 0$ for all $\mu$. Relations of this form can then be used to prove a self-testing statement according to \cref{def:selfTest}.

To formulate an SOS for a given $\bar{B}$, we can follow Ref.~\cite{BampsPironio}: consider a vector $\bm{R}=[R_{0},\ldots,R_{j},\ldots]^{\mathrm{T}}$ whose components are polynomials of $A_x$ and $B_y$. Here, we consider the case where each polynomial $R_j$ is linear, writing $P_{\mu} = \sum_{j} q_{\mu}^{j}R_{j}$ for some coefficients $\{q_\mu^j\}_j$. Then
\begin{align}
    \bar{B} &= \sum_{\mu} P_{\mu}^{\dagger}P_{\mu} \nonumber \\
    &= \sum_{jk}   R_{j}^{\dagger}\left(\sum_{\mu}\big(q_{\mu}^{j}\big)^{*} q_{\mu}^{k}\right)R_{k} \nonumber \\
    &= \sum_{jk}R_{j}^{\dagger}M^{jk}R_{k} = \bm{R}^{\dagger}M\bm{R}\,, \label{eq:SOS}
\end{align}
where $M$ is the Gram matrix of the set of vectors $\{\bm{q}^{j}\}$. Since $M$ is a Gram matrix, it is positive semidefinite by construction. If a matrix $\bm{M}$ is found that satisfies~\eqref{eq:SOS} then polynomials $P_{\mu}$ can be found by the matrix square root:
\begin{align}
    \bar{B} &= \bm{R}^{\dagger}M\bm{R} = \Big( \sqrt{M}\bm{R} \Big)^{\dagger} \Big( \sqrt{M} \bm{R} \Big). \label{eq:sosDecomp}
\end{align}
Since each entry of the vector $\sqrt{M}\bm{R}$ takes the form $[\sqrt{M}\bm{R}]_{\mu} = \sum_{k}[\sqrt{M}]_{\mu k}R_{k}$, we find that $P_{\mu} =[\sqrt{M}\bm{R}]_{\mu}$ provides the set of polynomials that satisfies \cref{eq:SOS}. 

Finding a decomposition for a given Bell operator $\bar{B}$ can be done via semidefinite programming (SDP)~\cite{BampsPironio} or the formal SOS techniques presented in~\cite{barizien2023}. Moreover, given a target quantum strategy tailored to some DI task, one can use the method of~\cite{barizien2023} to derive polynomial constraints satisfied by the target strategy, from which a custom Bell expression can be derived. Alternatively, one can take the target distribution and study numerical bounds on the von Neumann entropy via SDP relaxations~\cite{BrownDeviceIndependent2}. A Bell expression optimal for randomness can be extracted from the dual program, in the sense that its maximum violation certifies the amount of randomness of the target. One can then look for structure in this expression, and use the techniques of~\cite{BampsPironio} to get SOS decompositions. Both approaches can lead to a Bell expression and SOS decomposition from which a self-testing proof of the target strategy can be attempted, leading to expressions such as those in Proposition 1.

\subsection{Sufficiency of a unique qubit strategy using Jordan's lemma}

Informally, Jordan's lemma~\cite{Jordan} states that for two observables $A_{0}$ and $A_{1}$ on a Hilbert space $\mathcal{H}$, each with eigenvalues $\pm 1$, there exists a basis transformation such that both are simultaneously block diagonal with block size no greater than two. The Hilbert space decomposes into this block diagonal structure, and, by dilating where necessary, we can take each block to be a qubit system. This has already been used in the context of self-testing (see, e.g.,~\cite{BardynSelfTest,SekatskiBuilidngBlocks,ValcarceSelfTest}) and here we begin with a few general statements before attacking our particular case.

Jordan's lemma can be stated as follows.
\begin{lemma}[Jordan's lemma]
Let $\tilde{A}_{0}$ and $\tilde{A}_{1}$ be two binary observables on a Hilbert space $\mathcal{H}_{\tilde{Q}_{A}}$. Then there exists a basis in which $\tilde{A}_0$ and $\tilde{A}_1$ are block diagonal with block dimensions at most $2$. \label{lem:Jordan}
\end{lemma}
Applying \cref{lem:Jordan} to both Alice's and Bob's observables, we can write 
\begin{equation}
    \tilde{A}_{x} = \sum_{i} \ketbra{i}{i}_{F_{A}}\otimes A_{x}^{i}, \ \tilde{B}_{y} = \sum_{j} \ketbra{j}{j}_{F_{B}}\otimes B_{y}^{j}, \label{eq:genobs}
\end{equation}
where $A_x^i$ and $B_y^j$ are $2\times2$ (any $1\times1$ blocks allowed by \cref{lem:Jordan} can be extended to $2\times2$). We identify $\mathcal{H}_{\tilde{Q}_{A}} = \mathcal{H}_{F_{A}}\otimes\mathcal{H}_{Q_{A}}$ where $F_{A}$ is a system that flags the $2\times 2$ Jordan block, and $Q_{A}$ is a qubit system (similarly for $\mathcal{H}_{\tilde{Q}_{B}}$). A state then takes a generic form
\begin{equation}
    \rho_{\tilde{Q}_{A}\tilde{Q}_{B}} = \sum_{ii'jj'}\ketbra{i}{i'}_{F_{A}} \otimes \ketbra{j}{j'}_{F_{B}}\otimes\rho^{(i,i'),(j,j')}, \label{eq:veryGenState}
\end{equation}
where $\rho^{(i,i'),(j,j')}$ are $4\times 4$ matrices, $\mathrm{Tr}[\rho_{\tilde{Q}_{A}\tilde{Q}_{B}}] = \sum_{ij}\mathrm{Tr}[\rho^{(i,i),(j,j)}] = \sum_{ij}p_{ij}= 1$ and $\rho_{\tilde{Q}_{A}\tilde{Q}_{B}} \succeq 0$, where $\rho^{(i,i),(j,j)} = p_{ij}\rho^{ij}$ for a two-qubit density operator $\rho^{ij}$ and $p_{ij}>0$. Any shifted Bell operator $\bar{B}$ (of the form \cref{eq:BellSOS}) can then be written as
\begin{equation}
    \bar{B} = \sum_{ij}\ketbra{i}{i}_{F_{A}}\otimes\ketbra{j}{j}_{F_{B}}\otimes\bar{B}^{ij}, 
\end{equation}
where $\bar{B}^{ij}$ is constructed according to $\bar{B}$ with the qubit observables $\{A_{x}^{i}\}_x,\{B_{y}^{j}\}_y$. We therefore find $\langle \bar{B} \rangle = 0$ implies 
\begin{equation}
    \mathrm{Tr}[\rho^{ij}\bar{B}^{ij}] = 0 , \quad \forall\ i,j. \label{eq:maxDiag}
\end{equation}
The above equation tells us that, if a general strategy is maximally violating, every diagonal two-qubit block must also be maximally violating. This implies a system of constraints on every two-qubit strategy, which one can try to solve. Below we show solving this qubit problem is sufficient to determine the solution to the general problem. We proceed with a useful lemma.

\begin{lemma}\label{lem:zerocols}
    If $M\succeq0$ and $M_{ii}=0$ for some $i$ (i.e., one of the diagonal elements of $M$ is zero), then $M_{ij}=M_{ji}=0$ for all $j$ (i.e., the $i^{\mathrm{th}}$ row and column of $M$ are zero).
\end{lemma}
\begin{proof}
    If $R$ is a principal submatrix of $M$ (i.e., a submatrix containing the same rows of $M$ as columns), $M\succeq0$ implies $R\succeq0$. If $M_{ii}=0$ for some $i$, consider the principal submatrix $R=\left(\begin{array}{cccc}
 M_{ii} & M_{ij} \\
 M_{ij}^* & M_{jj} 
\end{array}\right)$. One of its eigenvalues is $\frac{1}{2}\left(M_{jj}-\sqrt{M_{jj}^2+4|M_{ij}|^2}\right)$, which cannot be positive, and is zero if and only if $M_{ij}=0$.
\end{proof}

\begin{lemma}\label{lem:blocklem}
    Let $d = 4n$ for some positive integer $n$, $\rho_{\psi} = \ketbra{\psi}{\psi}$ for some two-qubit pure state $\ket{\psi}$ and $\gamma_{\mu} > 0$ for $\mu \in \{0,...,n-1\}$. Let $M \in \mathbb{C}^{d \times d}$ be a Hermitian matrix, which, when written in terms of $4\times4$ blocks has the form 
    \begin{equation}
        M = \sum_{\mu=0}^{n-1} \sum_{\nu=0}^{n-1} \ketbra{\mu}{\nu}\otimes M_{\mu\nu},
    \end{equation}
    where $M_{\mu\nu} \in \mathbb{C}^{4\times 4}$, and the diagonal blocks satisfy $M_{\mu\mu} = \gamma_{\mu}\rho_{\psi}$. Then $M \succeq 0$ implies the off-diagonal blocks satisfy
    \begin{equation}
        M_{\mu\nu} = c_{\mu\nu} \rho_{\psi} \ \forall \mu,\nu,
    \end{equation}
    for some $c_{\mu\nu} \in \mathbb{C}$. Moreover, $M$ must take the form 
    \begin{equation}\label{eq:sig}
        M = \sigma\otimes\rho_{\psi},
    \end{equation}
    for some $\sigma\succeq0$. 
\end{lemma}
\begin{proof}
    Choose a basis $\{\ket{\psi_0},\ldots,\ket{\psi_3}\}$ for $M_{\mu\mu}$ such that $\ket{\psi_0}=\ket{\psi}$. Using
    Lemma~\ref{lem:zerocols}, it follows that $M_{\mu\nu}=c_{\mu\nu}\ketbra{\psi}{\psi}$ with $c_{\mu\nu}\in\mathbb{C}$ for all $\mu,\nu$ (and $c_{\mu\mu}=\gamma_{\mu}$).
    
    Finally, we can write 
    \begin{equation}
        M = \sum_{\mu=0}^{n-1}\sum_{\nu=0}^{n-1}\ketbra{\mu}{\nu}\otimes c_{\mu\nu}\rho_{\psi} = \Bigg(  \sum_{\mu=0}^{n-1}\sum_{\nu=0}^{n-1}c_{\mu\nu} \ketbra{\mu}{\nu} \Bigg) \otimes \rho_{\psi},
    \end{equation}
    which is equivalent to the form~\eqref{eq:sig} taking $\sigma = \sum_{\mu=0}^{n-1}\sum_{\nu=0}^{n-1}c_{\mu\nu} \ketbra{\mu}{\nu}$. Note that because $M\succeq0$ and $\rho_\psi\succeq0$, we must have $\sigma\succeq0$.
\end{proof}

Now we present the main claim of this section. Informally, this says that in the case of two inputs and two outputs per party, it suffices to solve the self-testing strategy for systems that comprise 2 qubits, in order to solve the self-testing problem in the general case.

\begin{lemma}\label{lem:blocks}
    Let $B$ be a bipartite Bell operator in the two-input, two-output scenario with a quantum bound $\eta^{\mathrm{Q}}$. Suppose, for every two-qubit strategy $(\tilde{\rho}^{\mathsf{q}}_{Q_{A}Q_{B}},\tilde{A}_{x}^{\mathsf{q}},\tilde{B}_{y}^{\mathsf{q}})$ achieving $\langle B \rangle = \eta^{\mathrm{Q}}$, there exist unitaries $U_{A}$ on $\mathcal{H}_{Q_{A}}$ and $U_{B}$ on $\mathcal{H}_{Q_{B}}$ such that $(U_A\otimes U_B)\tilde{\rho}^{\mathsf{q}}_{Q_{A}Q_{B}}(U_A^\dagger\otimes U_B^\dagger) = \ketbra{\psi}{\psi}$,
    \begin{equation}
        U_{A} \tilde{A}^{\mathsf{q}}_{x}U_{A}^{\dagger} = A_{x},\quad\text{and}\quad U_{B} \tilde{B}^{\mathsf{q}}_{y}U_{B}^{\dagger} = B_{y},
    \end{equation}
    for a target strategy $(\ket{\psi},\{A_x\}_x,\{B_y\}_y)$. Then, for every quantum strategy $(\rho_{\tilde{Q}_{A}\tilde{Q}_{B}},\{\tilde{A}_{x}\}_x,\{\tilde{B}_{y}\}_y)$ achieving $\langle B \rangle = \eta^{\mathrm{Q}}$, there exist unitaries $V_{A}$ on $\mathcal{H}_{\tilde{Q}_{A}}$ and $V_{B}$ on $\mathcal{H}_{\tilde{Q}_{B}}$, and a state $\ket{\xi}_{\mathrm{Junk}} \in \mathcal{H}_{\mathrm{Junk}}$, such that for any purification $\ket{\Psi}_{E\tilde{Q}_{A}\tilde{Q}_{B}}$ of $\rho_{\tilde{Q}_{A}\tilde{Q}_{B}}$
    \begin{equation} 
        (\mathbb{I}_{E}\otimes V_{A} \otimes V_{B}) (\mathbb{I}_{E}\otimes\tilde{A}_{x}\otimes\tilde{B}_{y}) \ket{\Psi}_{E\tilde{Q}_{A}\tilde{Q}_{B}} = \ket{\xi}_{\mathrm{Junk}}\otimes(A_{x} \otimes B_{y})\ket{\psi}_{Q_{A}Q_{B}}\ \forall\ x,y.
      \end{equation}
    \end{lemma}
\begin{proof}
  We begin by applying Jordan's lemma, and write the observables $\tilde{A}_{x},\tilde{B}_{y}$ in the form of \cref{eq:genobs}. Let $\rho := \mathrm{Tr}_{E}[\ketbra{\Psi}{\Psi}]$ be written in the generic form of \cref{eq:veryGenState}. Then the fact that every diagonal two-qubit block must maximally violate the Bell inequality (see \cref{eq:maxDiag}), combined with the uniqueness assumption of qubit strategies made in the lemma implies the existence of unitaries $U_{A}^{ij}$, $U_{B}^{ij}$ for every diagonal block $(i,j)$ satisfying 
    \begin{equation}
         (U_A^{ij}\otimes U_B^{ij})\rho^{ij}((U_A^{ij})^\dagger\otimes(U_B^{ij})^\dagger) = \ketbra{\psi}{\psi}, \ U_{A}^{ij} A^{i}_{x}(U_{A}^{ij})^\dagger = A_{x}, \ U_{B}^{ij} B^{j}_{y}(U_{B}^{ij})^\dagger = B_{y}\ \forall\ x,y,i,j. \label{eq:qubitRel}
    \end{equation}
    We now need to argue that these unitaries only vary with the local block index of each party, i.e., $U_{A}^{ij} \equiv U_{A}^{i}$ and $U_{B}^{ij} \equiv U_{B}^{j}$. To do so we use the result of the following lemmas.
    \begin{lemma} \label{lem:uptodiag}
Let $E$ and $F$ be two Hermitian operators with non-degenerate eigenspaces such that $E=UFU^\dagger$ for some unitary $U$. Then if $E=VFV^\dagger$ for a unitary $V$, then $V=UD$, where $D$ is diagonal in the eigenbasis of $F$.
  \end{lemma}
\begin{proof}
    $UFU^\dagger=VFV^\dagger$ rearranges to $[V^\dagger U,F]=0$, from which it follows that $V^\dagger U$ is diagonal in the eigenbasis of $F$, i.e., we can write $V=UD$ with $D$ diagonal in the eigenbasis of $F$.
\end{proof}
\begin{lemma}\label{lem:com}  
For $i=0$ and $i=1$, let $E_i$ and $F_i$ be Hermitian operators with non-degenerate eigenspaces such that $E_i=UF_iU^\dagger$ for some unitary $U$, where no eigenstate of $F_0$ is orthogonal to an eigenstate of $F_1$. Then if there exists a unitary $V$ such that $E_i=VF_iV^\dagger$ for $i=0$ and $i=1$, $V$ and $U$ are equal up to a phase.
  \end{lemma}
  \begin{proof}
    Using Lemma~\ref{lem:uptodiag}, for $E_0$ and $F_0$ we can conclude $V=UD_0$, where $D_0$ is diagonal in the eigenbasis of $F_0$.  Similarly, we can conclude that $V=UD_1$, where $D_1$ is diagonal in the eigenbasis of $F_1$. These imply $D_0=D_1$ and hence $D_0$ must be diagonal in the eigenbasis of both $F_0$ and $F_1$. Thus, $D_0=\sum_ie^{\mathrm{i}\alpha_i}\ketbra{i}{i}=\sum_ie^{\mathrm{i}\beta_i}\ketbra{\phi_i}{\phi_i}$, where $\{\ket{i}\}$ and $\{\ket{\phi_i}\}$ are orthonormal eigenbases for $F_0$ and $F_1$. We have
\begin{equation}\label{eq:oneeq}
1=e^{-\mathrm{i}\alpha_i}\bra{i}D_0\ket{i}=\sum_je^{\mathrm{i}(\beta_j-\alpha_i)}|\braket{i}{\phi_j}|^2.
\end{equation}
By assumption, $|\braket{i}{\phi_j}|^2\neq0$ for any $i,j$. Hence~\eqref{eq:oneeq} can only hold if $e^{\mathrm{i}(\beta_j-\alpha_i)}=1$ for all $i,j$, i.e., if $D_1=e^{\mathrm{i}\alpha}\mathbb{I}$ for some $\alpha\in\mathbb{R}$. This implies $V=e^{\mathrm{i}\alpha}U$.
\end{proof}
We can apply \cref{lem:com} to the unitaries $U_{A}^{ij},U_{B}^{ij}$. Let $i \neq i'$, then we know $U_{B}^{ij}B_{y}^{j}(U_{B}^{ij})^\dagger = B_{y}$ and $U_{B}^{i'j}B_{y}^{j}(U_{B}^{i'j})^\dagger = B_{y}$. No eigenstate of $B_{0}^{j}$ can be orthogonal to an eigenstate of $B_{1}^{j}$ (otherwise they would commute and there would be no classical-quantum gap for that block, contradicting \cref{eq:maxDiag}), so applying \cref{lem:com} with $F_{0} = B_{0}^{j}$, $F_{1} = B_{1}^{j}$, $U = U_{B}^{ij}$ and $V=U_{B}^{i'j}$ we find $U_{B}^{ij} = \exp(\mathrm{i}\zeta_{i,i'}^{j})U_{B}^{i'j}$ for some phase $\zeta_{i,i'}^{j}$. By a similar argument $U_{A}^{ij} = \exp(\mathrm{i}\xi_{j,j'}^{i})U_{A}^{ij'}$ for $j \neq j'$. Noting that $U\mapsto e^{\mathrm{i}\phi}U$ preserves $UGU^\dagger$ for any operator $G$, we can drop these phases and define the block diagonal unitaries
    \begin{equation}
        \begin{aligned}
            V_{A} &:= \sum_{i} \ketbra{i}{i} \otimes U_{A}^{i,0} , \\
            V_{B} &:= \sum_{j} \ketbra{j}{j} \otimes U_{B}^{0,j}.
        \end{aligned}
    \end{equation}
    We emphasise that $V_{A},V_{B}$ are local unitaries, since they vary only with the local block index of each party. We have that 
    \begin{equation}
    \begin{aligned}
        \tau&:=(V_A\otimes V_B)\rho_{\tilde{Q}_{A}\tilde{Q}_{B}}(V_A^\dagger\otimes V_B^\dagger) \\
        &= \sum_{ii'jj'} \ketbra{i}{i'} \otimes \ketbra{j}{j'} \otimes (U_{A}^{i,0} \otimes U_{B}^{0,j}) \rho^{(i,i'),(j,j')} (U_{A}^{i',0} \otimes U_{B}^{0,j'})^{\dagger}.
    \end{aligned}
    \end{equation}
     In particular, the diagonal blocks of $\tau$ ($i=i'$ and $j=j'$) are of the form $\gamma_{i,j}\ketbra{\psi}{\psi}$. This is the form required for Lemma~\ref{lem:blocklem} up to a change of notation (each pair of values of $(i,j)$ corresponds to a value of $\mu$).  Lemma~\ref{lem:blocklem} gives that $\tau=\sigma\otimes\ketbra{\psi}{\psi}$ where $\sigma$ is some density operator on $\mathcal{H}_{F_A}\otimes\mathcal{H}_{F_B}$. It follows that any purification of $\rho_{\tilde{Q}_{A}\tilde{Q}_{B}}$ takes the form
    \begin{equation}
        \ket{\Psi}_{E\tilde{Q}_{A}\tilde{Q}_{B}} = V^{\dagger}\Big[\ket{\xi}_{EF_{A}F_{B}}\otimes\ket{\psi}_{Q_{A}Q_{B}}\Big],
    \end{equation}
    where $V=\mathbb{I}_E\otimes V_A\otimes V_B$, and $\ket{\xi}_{EF_AF_B}$ is a purification of $\sigma$. For the observables, we have
    \begin{equation}
        V_{A}\tilde{A}_{x}V_{A}^{\dagger} = \sum_{i} \ketbra{i}{i} \otimes U_{A}^{i,0}A_{x}^{i}(U_{A}^{i,0})^{\dagger} = \sum_{i} \ketbra{i}{i} \otimes A_{x} = \mathbb{I}_{F_{A}}\otimes A_{x},
    \end{equation}
    where we applied \cref{eq:qubitRel}, and similarly for $V_{B}\tilde{B}_{y}V_{B}^{\dagger}$. From this it follows that
    \begin{align}
        V(\mathbb{I}_{E}\otimes\tilde{A}_x\otimes\tilde{B}_y)V^{\dagger}V\ket{\Psi}_{E\tilde{Q}_{A}\tilde{Q}_{B}}=\ket{\xi}_{EF_{A}F_{B}}\otimes(A_x\otimes B_y)\ket{\psi}_{Q_{A}Q_{B}}\ \forall\ x,y,
    \end{align}
    which is the general self-testing statement with the junk space identified with $EF_AF_B$. 
\end{proof}

\subsection{Further two-qubit reductions for Bell inequalities without marginal terms}
In this subsection, we show how the qubit strategy reductions presented in~\cite{PABGMS,Bhavsar2023} can also be applied in the self-testing scenario. Because of Lemma~\ref{lem:blocks}, it suffices to consider the two-qubit case. We define the following Bell basis for later use: $\{\ket{\Phi_{0}},\ket{\Phi_{1}},\ket{\Phi_{2}},\ket{\Phi_{3}}\}$ where 
\begin{align}
    \ket{\Phi_{0}}&= \frac{1}{\sqrt{2}}(\ket{00} + \ket{11}), \nonumber \\
    \ket{\Phi_{1}}&= \frac{1}{\sqrt{2}}(\ket{00} - \ket{11}), \nonumber \\
    \ket{\Phi_{2}}&= \frac{1}{\sqrt{2}}(\ket{01} + \ket{10}), \nonumber \\
    \ket{\Phi_{3}}&= \frac{1}{\sqrt{2}}(\ket{01} - \ket{10}),
\end{align}
and write $\Phi_{\alpha} = \ketbra{\Phi_{\alpha}}{\Phi_{\alpha}}$.

In \cref{lem:reduce1}, we show that for every two-qubit strategy which maximally violates the Bell inequality, we can construct another strategy with reduced parameters, which also maximally violates the Bell inequality. Note, the operations used to construct this reduced strategy are not allowed under the self-testing definition. Instead, we are concerned with uniqueness; in \cref{lem:reduce2} we will prove that a unique form of the reduced strategy implies a unique form of the general strategy from which it was constructed (all up to local unitaries), when the target state is maximally entangled.     

\begin{lemma}
    Let $\rho$ be a two-qubit state, and $A_x$ and $B_y$ be qubit observables with eigenvalues $\pm 1$ for $x,y\in\{0,1\}$. Let $P(A_x,B_y)$ be a linear function of $\{\mathbb{I}\}\cup\{A_xB_y\}_{x,y}$ with real coefficients satisfying $\mathrm{Tr}[\rho P(A_x,B_y)] = 0$. Then there exists another two-qubit state $\rho'$, and observables $A'_x$, $B'_y$ satisfying $\mathrm{Tr}[\rho' P(A'_x,B'_y)] = 0$, where 
    \begin{equation}
    \begin{aligned}
        \rho' &= \sum_{\alpha=0}^{3}\lambda_{\alpha}'\ketbra{\Phi_{\alpha}}{\Phi_{\alpha}}, \\
        A'_x &= \cos(a_x')\,\sigma_{Z} + \sin(a_x')\, \sigma_{X}, \\
        B'_y &= \cos(b_y')\,\sigma_{Z} + \sin(b_y')\, \sigma_{X}, 
    \end{aligned} \label{eq:red}
    \end{equation}
    for some $\lambda_{\alpha}' \geq 0$, $\sum_{\alpha}\lambda_{\alpha}' = 1$ and $a_x',b_y' \in \mathbb{R}$ for all $x,y$. \label{lem:reduce1}
\end{lemma}
\begin{proof}
  Without loss of generality, we can assume each party rotates their local basis such that
\begin{align}
        A_{x} = \cos a_{x} \, \sigma_{Z} + \sin a_{x} \, \sigma_{X}, \nonumber \\
        B_{y} = \cos b_{y} \, \sigma_{Z} + \sin b_{y} \, \sigma_{X}.\label{eq:mmts}
\end{align}
We further suppose Alice rotates such that $A_0=\sigma_Z$, i.e., $a_0=0$. Note that Bob could also rotate such that $B_0=\sigma_Z$ but we do not apply this and keep $b_0$ and $b_1$ unconstrained. The state remains arbitrary. Given the form of the polynomial $P$, we have $(\sigma_{Y} \otimes \sigma_{Y})P(A_{x},B_{y})(\sigma_{Y} \otimes \sigma_{Y}) = P(A_{x},B_{y})$ (note this follows from $\sigma_{Y}O\sigma_{Y} = - O$ for any observable $O = \cos(t) \, \sigma_{Z} + \sin(t) \, \sigma_{X}, \ t \in \mathbb{R}$).

Given a state $\rho$ and defining
    \begin{equation}
        \bar{\rho} := \frac{\rho + (\sigma_{Y} \otimes \sigma_{Y})\rho(\sigma_{Y} \otimes \sigma_{Y})}{2}, \label{eq:sym}
    \end{equation}
    we have $\mathrm{Tr}[\bar{\rho}P(A_{x},B_{y})]=\mathrm{Tr}[\rho P(A_{x},B_{y})]=0$. References~\cite{PABGMS,Bhavsar2023} showed $\bar{\rho}$ has the form
    \begin{equation}
        \bar{\rho} = \begin{bmatrix}
            \lambda_{0} & 0 & 0 & r\\
            0 & \lambda_{1} & s & 0 \\
            0 & s^{*} &\lambda_{2} & 0 \\
            r^{*} & 0 & 0 & \lambda_{3}
        \end{bmatrix} \label{eq:genState}
    \end{equation}
    in the basis $\{\ket{\Phi_{0}},\ket{\Phi_{1}},\ket{\Phi_{2}},\ket{\Phi_{3}}\}$, where $r,s \in \mathbb{C}$, $\lambda_{\alpha} \geq 0$, $\sum_{\alpha}\lambda_{\alpha} = 1$ and positivity implies $|s|^{2} \leq \lambda_{1}\lambda_{2}$, $|r|^{2} \leq \lambda_{0}\lambda_{3}$. Consider the state $\bar{\rho}^*$ formed by taking the conjugate of $\bar{\rho}$ in the given Bell basis, and notice 
    \begin{equation}
        \mathrm{Tr}[\bar{\rho}P(A_{x},B_{y})] = \mathrm{Tr}[\bar{\rho}P(A_{x},B_{y})]^{*} = \mathrm{Tr}[\bar{\rho}^{*}P(A_{x},B_{y})^{*}] = \mathrm{Tr}[\bar{\rho}^{*}P(A_{x},B_{y})] = 0,
    \end{equation}
    where we used the fact that $P(A_{x},B_{y})$ has real coefficients, and that each of the four operators $\sigma_{X/Z}\otimes\sigma_{X/Z}$ is real in the Bell basis, so, when taking the conjugate in the Bell basis $P(A_x,B_y)^*=P(A_x,B_y)$. It follows that if we define
    \begin{equation}
        \bar{\bar{\rho}} := \frac{\bar{\rho} + \bar{\rho}^*}{2}, \label{eq:cc}
    \end{equation}
    then $\mathrm{Tr}[\bar{\bar{\rho}}P(A_{x},B_{y})] = 0$.
    References~\cite{PABGMS,Bhavsar2023} showed for any state of the form $\bar{\bar{\rho}}$, there exist local unitaries $U_A$ and $U_B$ such that $\rho':=(U_A\otimes U_B)\bar{\bar{\rho}}(U_A^\dagger\otimes U_B^\dagger)$ is diagonal in the Bell basis, and with the property that $U_A$ and $U_B$ maintain the form of the observables (i.e., they still take the form~\eqref{eq:mmts}, just with different angles).  In other words, after applying these unitaries, 
    \begin{equation}
        \rho' = \sum_{\alpha=0}^{3}\lambda_{\alpha}'\ketbra{\Phi_\alpha}{\Phi_\alpha}, \label{eq:bellDiag}
    \end{equation}
    and we have a new set of rotated measurements $A_{x}',B_{y}'$ with modified angles $a_{x}',b_{y}'$.
    Finally, note that
    \begin{equation}
        \mathrm{Tr}[\bar{\bar{\rho}}P(A_x,B_y)]  
        = \mathrm{Tr}[\rho'(U_A\otimes U_B)P(A_x,B_y)(U_A^\dagger\otimes U_B^\dagger)] 
        = \mathrm{Tr}[\rho'P(A_x',B_y')] = 0\,.
\end{equation}
This proves the claim.
\end{proof}

To complete this section we show how uniqueness of the reduced two-qubit strategy implies uniqueness of the generic two-qubit strategy when the target strategy involves a maximally entangled state.

\begin{lemma}\label{lem:reduce2}
    Let $\bar{B}$ be any Bell operator in the bipartite two-input two-output scenario without marginal terms. Suppose every reduced strategy of the form \cref{eq:red} that achieves $\langle \bar{B} \rangle = 0$ is given by $\lambda_{0}' = 1$, $\lambda_{\alpha}' = 0$ for $\alpha \neq 0$ and $(a_{x}',b_{y}') = (\alpha_{x},\beta_{y})$ for some fixed angles $\{\alpha_x\}_x$ and $\{\beta_y\}_y$, up to local unitaries that maintain the form~\eqref{eq:red}. Then, up to local unitaries, every two-qubit strategy that achieves $\langle \bar{B} \rangle = 0$ is given by 
        \begin{equation}
        \tilde{\rho}^{\mathsf{q}} = \Phi_{0}, \  \tilde{A}_{x}^{\mathsf{q}} = \cos(\alpha_{x}) \, \sigma_{Z} + \sin(\alpha_{x}) \, \sigma_{X}, \ \tilde{B}_{y}^{\mathsf{q}} = \cos(\beta_{y}) \, \sigma_{Z} + \sin(\beta_{y}) \, \sigma_{X}. \label{eq:red2}
    \end{equation}
\end{lemma}
\begin{proof}
Consider a strategy of the form~\eqref{eq:red} with $\lambda_{0}' = 1$, $\lambda_{\alpha}' = 0$ for $\alpha \neq 0$ and $(a_{x}',b_{y}') = (\alpha_{x},\beta_{y})$. Tracing back the construction in the proof of Lemma~\ref{lem:reduce1}, we have that $\rho' = \Phi_{0}$ implies $\bar{\bar{\rho}} = \hat{\Phi}_{0}$, $\bar{\rho} = \hat{\Phi}_{0}$ and $\rho = \hat{\Phi}_{0}$, where $\hat{\Phi}_{0} = (U_{A}^{\dagger}\otimes U_{B}^{\dagger})\Phi_{0}(U_{A} \otimes U_{B})$\footnote{This follows from the fact that $\alpha\rho+(1-\alpha)\sigma$ can only be pure if $\rho=\sigma$.}. Therefore, any initial strategy $(\rho,\{A_{x}\}_{x},\{B_{y}\}_{y})$ for which $\langle \bar{B} \rangle = 0$ must be of the form $(\hat{\Phi}_{0},\{A_{x}\}_{x},\{B_{y}\}_{y})$. If Alice and Bob apply the local unitaries $U_{A},U_{B}$, we obtain $(\Phi_{0},\{A_{x}'\}_{x},\{B_{y}'\}_{y})$ where $A_{x}',B_{y}'$ are given by the angles $\alpha_{x},\beta_{y}$.
\end{proof}

\section{Proof of Proposition 1} \label{app:selfProof}

\noindent \textbf{Proposition 1.} \textit{Let $\theta,\phi,\omega \in \mathbb{R}$. Define the family of Bell expressions $\langle B_{\theta,\phi,\omega} \rangle$,
\begin{multline}
    \langle B_{\theta,\phi,\omega} \rangle = \cos(\theta + \phi)\cos(\theta+\omega)\, \big \langle A_{0}\big( \cos \omega \, B_{0} - \cos \phi \, B_{1} \big) \big \rangle    \\ +  \cos \phi \cos \omega \, \big \langle  A_{1}\big( -\cos(\theta + \omega) \, B_{0} + \cos (\theta + \phi) \, B_{1} \big) \big \rangle. 
\end{multline}
Then the following hold:
\begin{enumerate}[(i)]
    \item The local bounds are given by $\pm \eta^{\mathrm{L}}_{\theta,\phi,\omega}$, where $ \eta^{\mathrm{L}}_{\theta,\phi,\omega} = \max_{\pm} \Big\{\left| \cos(\theta+\omega)\cos(\omega)\big( \cos(\theta+\phi)\pm\cos(\phi) \big)\right|+\left| \cos(\theta+\phi)\cos(\phi)\big( \cos(\theta+\omega)\pm\cos(\omega) \big)\right| \Big\}.$
    \item If 
\begin{equation}
\cos(\theta+\phi)\cos(\phi)\cos(\theta+\omega)\cos(\omega)<0, \ifarxiv\tag{\ref{cond}}\else\fi\label{eq:cond1App}
\end{equation}
then the quantum bounds are given by $\pm \eta^{\mathrm{Q}}_{\theta,\phi,\omega}$, where $\eta^{\mathrm{Q}}_{\theta,\phi,\omega} = \sin(\theta)\sin(\omega-\phi)\sin(\theta+\omega+\phi)$.
\item $\left|\eta^{\mathrm{Q}}_{\theta,\phi,\omega}\right|> \eta^{\mathrm{L}}_{\theta,\phi,\omega}\iff\cos(\theta+\phi)\cos(\phi)\cos(\theta+\omega)\cos(\omega)<0.$
    \item 
    If~\eqref{eq:cond1App} holds, up to local isometries there is a unique strategy that achieves $\langle B_{\theta,\phi,\omega} \rangle = B_{\theta,\phi,\omega}^{\mathrm{Q}}$: 
    \begin{gather}
        \rho_{Q_{A}Q_{B}} = \ketbra{\psi}{\psi}, \ \mathrm{where} \ \ket{\psi} = \frac{\ket{00} + i\ket{11}}{\sqrt{2}}, \nonumber \\
        A_{0} = \sigma_{X}, \
        A_{1} = \cos \theta \, \sigma_{X} + \sin \theta \, \sigma_{Y}, \nonumber \\
        B_{0} = \cos \phi \, \sigma_{X} + \sin \phi \, \sigma_{Y}, \
        B_{1} = \cos \omega \, \sigma_{X} + \sin \omega \, \sigma_{Y}. \label{eq:appstrat}
    \end{gather}
\end{enumerate}}

\vspace{0.2cm}

\begin{proof}
{\bf Part (i)}\\
\noindent To prove the the above statement, we follow the same steps taken in~\cite{WBC}. The maximum local value of the Bell expression $\langle B_{\theta,\phi,\omega} \rangle$ can be found by optimizing over local deterministic strategies, obtained by setting $\langle A_{x}B_{y} \rangle=a_{x}b_{y}$ for $a_{x},b_{y}\in \{\pm1\}$,
\begin{multline}  \eta^{\mathrm{L}}_{\theta,\phi,\omega} = \max_{a_{x},b_{y} \in \{\pm 1\}}\Big\{\cos(\theta+\omega)\cos(\omega)\big(a_0\cos(\theta+\phi)-a_1\cos(\phi)\big)b_0\\-\cos(\theta+\phi)\cos(\phi)\big(a_0\cos(\theta+\omega)-a_1\cos(\omega)\big)b_1\Big\}.
\end{multline}
To maximize the above expression, we want
\begin{align}
    b_0&=\mathrm{Sgn}\Big[\cos(\theta+\omega)\cos(\omega)\big(a_0\cos(\theta+\phi)-a_1\cos(\phi)\big)\Big], \\
b_1&=\mathrm{Sgn}\Big[-\cos(\theta+\phi)\cos(\phi)\big(a_0\cos(\theta+\omega)-a_1\cos(\omega)\big)\Big],
\end{align}
where $\mathrm{Sgn}[\cdot]$ is the sign function. We are now left with two choices: either $a_0=a_1$, or $a_0=-a_1$, which leaves the following maximum value
\begin{multline}
\eta^{\mathrm{L}}_{\theta,\phi,\omega} = \max_{\pm} \Big\{\left| \cos(\theta+\omega)\cos(\omega)\big( \cos(\theta+\phi)\pm\cos(\phi) \big)\right|+\left| \cos(\theta+\phi)\cos(\phi)\big( \cos(\theta+\omega)\pm\cos(\omega) \big)\right| \Big\}. 
\end{multline}
By switching the sign of $b_0$ and $b_1$, we also obtain minimum value $-\eta^{\mathrm{L}}_{\theta,\phi,\omega}$. This concludes part $(i)$.
\bigskip

{\bf Part (ii)}\\
\noindent For part $(ii)$ and $(iv)$, we use the following lemma, which provides the SOS decomposition found in~\cite{barizien2023} for the shifted Bell operator $\bar{B}_{\theta,\phi,\omega} = \eta^{\mathrm{Q}}_{\theta,\phi,\omega}\mathbb{I} - B_{\theta,\phi,\omega}$.
\begin{lemma}[\cite{barizien2023}] \label{lem:SOS}
    Let us define the polynomials $\{R_{0},R_{1}\}$, where
    \begin{align}
    R_{0} &= \sin \theta \, B_{0} + \cos(\theta + \phi) A_{0} - \cos \phi\, A_{1}, \\
    R_{1} &= \sin \theta \, B_{1} + \cos (\theta+\omega) A_{0} - \cos \omega\, A_{1}.
\end{align}
Then the shifted Bell operator $\bar{B}_{\theta,\phi,\omega}$ has the decomposition,
\begin{equation}
    \bar{B}_{\theta,\phi,\omega} = c_{0}R_{0}^{\dagger}R_{0} + c_{1}R_{1}^{\dagger}R_{1}, \label{eq:SOS1}
\end{equation}
where
\begin{align}
    c_{0} = \frac{-\cos \omega \cos(\theta+\omega)}{2\sin \theta},  \  
    c_{1} = \frac{ \cos \phi \cos(\theta+\phi)}{2\sin \theta}.
\end{align}
\end{lemma}
\begin{proof}
The equality in \cref{eq:SOS1} can be verified by direct calculation, noting that $A_0^2=A_1^2=B_0^2=B_1^2=\id$. To ensure one of $\pm\bar{B}_{\theta,\phi,\omega} \succeq 0$, we need $\mathrm{sgn}(c_{0}) = \mathrm{sgn}(c_{1})$, or equivalently $c_{0}c_{1} > 0$, which is verified by noting, for the range of parameters satisfying \cref{eq:cond1App}, 
\begin{equation*}
    c_{0}c_{1} = \frac{-\cos(\theta+\phi)\cos(\phi)\cos(\theta+\omega)\cos(\omega)}{4 \sin^{2}(\theta)} > 0.\qedhere
\end{equation*}
\end{proof}
Cast into the Gram matrix representation described earlier in this section, we find $\bm{R} = [R_{0},R_{1}]^{\mathrm{T}}$, and
\begin{equation}
    \bm{M}= \frac{1}{2\sin \theta}\begin{bmatrix}
    -\cos \omega \cos(\theta + \omega) & 0 \\
    0 & \cos \phi \cos(\theta+\phi)
    \end{bmatrix}.
\end{equation}
Condition 1 in \cref{eq:cond1App} then translates either to $\bm{M}$ or $-\bm{M}$ being positive definite ($\succ 0$), and the SOS polynomials are given either by the entries of the vector $\sqrt{\bm{M}}\bm{R}$ or $\sqrt{-\bm{M}}\bm{R}$, given by $P_{0} = \sqrt{|c_{0}|}R_{0}$ $P_{1} = \sqrt{|c_{1}|}R_{1}$. \cref{lem:SOS} proves $\pm \bar{B}_{\theta,\phi,\omega} \succeq 0$ when $\pm M \succ 0$ respectively, and since $\langle \bar{B}_{\theta,\phi,\omega} \rangle = 0$ is achievable by the strategy in \cref{eq:appstrat}, this proves Part $(ii)$.
\bigskip

{\bf Part (iii)}\\
\noindent For Part $(iii)$, note that $\max\{f,g\}< h\iff f<h$ and $g<h$, and furthermore that
\begin{equation}\label{eq:mods}
|f|+|g|<|h|\iff2|f||g|<|h|^2-|f|^2-|g|^2\,.
\end{equation}
That the quantum maximum is always greater than the classical one is equivalent to $\max\{|f_1|+|g_1|,|f_2|+|g_2|\}<|h|$, where
\begin{align*}
f_1&=\cos(\theta+\omega)\cos(\omega)\big( \cos(\theta+\phi)+\cos(\phi) \big),\\ g_1&=\cos(\theta+\phi)\cos(\phi)\big( \cos(\theta+\omega)+\cos(\omega) \big),\\ f_2&=\cos(\theta+\omega)\cos(\omega)\big( \cos(\theta+\phi)-\cos(\phi) \big),\\ g_2&=\cos(\theta+\phi)\cos(\phi)\big( \cos(\theta+\omega)-\cos(\omega) \big)\text{ and}\\
h&=\sin(\theta)\sin(\omega-\phi)\sin(\theta+\omega+\phi)\,.
\end{align*}

Using~\eqref{eq:mods}, the expression $|f_1|+|g_1|<|h|$ is equivalent to
\begin{align}\label{eq:mods1}
|f_1g_1|<-\cos(\theta+\phi)\cos(\phi)\cos(\theta+\omega)\cos(\omega)(2+\cos(\theta+2\phi)+\cos(\theta+2\omega))\cos^2(\theta/2)\,.
\end{align}
For this to hold we require
\begin{equation}\label{eq:req}
  \cos(\theta+\phi)\cos(\phi)\cos(\theta+\omega)\cos(\omega)<0\,,
\end{equation}
which is~\eqref{eq:cond1App}.

We have
\begin{align*}  |f_1g_1|&=\left|\cos(\theta+\phi)\cos(\phi)\cos(\theta+\omega)\cos(\omega)(\cos(\theta+\phi)+\cos(\phi))(\cos(\theta+\omega)+\cos(\omega))\right|\\
  &=-\cos(\theta+\phi)\cos(\phi)\cos(\theta+\omega)\cos(\omega)\left|(\cos(\theta+\phi)+\cos(\phi))(\cos(\theta+\omega)+\cos(\omega))\right|,
\end{align*}
and hence, using~\eqref{eq:req}, we can reduce~\eqref{eq:mods1} to
\begin{align} \left|(\cos(\theta+\phi)+\cos(\phi))(\cos(\theta+\omega)+\cos(\omega))\right|<(2+\cos(\theta+2\phi)+\cos(\theta+2\omega))\cos^2(\theta/2)\,.
\end{align}
Squaring both sides and rearranging we find that this is equivalent to
\begin{align}\label{eq:pos1}
\cos^4(\theta/2)\sin^2(\omega-\phi)\sin^2(\theta+\omega+\phi)>0\,,
\end{align}
i.e., none of the terms in the product on the left hand side of~\eqref{eq:pos1} can be zero.

Similarly, replacing $f_1$ and $g_1$ with $f_2$ and $g_2$ gives
\begin{align*}
|f_2g_2|=\left|\cos(\theta+\phi)\cos(\phi)\cos(\theta+\omega)\cos(\omega)(\cos(\theta+\phi)-\cos(\phi))(\cos(\theta+\omega)-\cos(\omega))\right|,
\end{align*}
and hence (again under the condition~\eqref{eq:req}) we can reduce to
\begin{align*}
\left|(\cos(\theta+\phi)-\cos(\phi))(\cos(\theta+\omega)-\cos(\omega))\right|<(2-\cos(\theta+2\phi)-\cos(\theta+2\omega))\sin^2(\theta/2)\,.
\end{align*}
Squaring both sides and rearranging we find that this is equivalent to
\begin{align}\label{eq:pos2}
\sin^4(\theta/2)\sin^2(\omega-\phi)\sin^2(\theta+\omega+\phi)>0\,.
\end{align}

Hence, $\max\{|f_1|+|g_1|,|f_2|+|g_2|\}<|h|$ implies~\eqref{eq:req}, \eqref{eq:pos1} and~\eqref{eq:pos2}. That~\eqref{eq:req}, \eqref{eq:pos1} and~\eqref{eq:pos2} imply $\max\{|f_1|+|g_1|,|f_2|+|g_2|\}<|h|$ can also be seen following the argument above.

Finally, we note that the relations~\eqref{eq:pos1} and~\eqref{eq:pos2} are equivalent to
\begin{equation}
\sin(\theta)\sin(\omega-\phi)\sin(\theta+\omega+\phi)\neq0\,,
\end{equation}
and that, in fact, this is implied by~\eqref{eq:req}.  To establish the latter, note that $\cos(\theta+\phi)\cos(\phi)\cos(\theta+\omega)\cos(\omega)$ can be rewritten as either
\begin{align*}
    &\cos^2(\phi)\cos^2(\theta)\cos^2(\omega)-\cos(\phi)\cos(\omega)\sin(\theta)\left(\sin(\theta+\omega+\phi)-\sin(\theta)\cos(\phi)\cos(\omega)\right)\quad\text{or}\\
    &\cos^2(\omega)\cos^2(\theta+\phi)-\sin(\theta)\cos(\omega)\cos(\theta+\phi)\sin(\omega-\phi).
\end{align*}
Then substituting $\sin(\theta+\omega+\phi)=0$ into the first, $\sin(\omega-\phi)=0$ into the second and $\sin(\theta)=0$ into either gives an expression that cannot be smaller than $0$.  This establishes Part~$(iii)$.
\bigskip

{\bf Part (iv)}\\
We use the results of \cref{sec:jlem}, where we showed in \cref{lem:blocks} that proving there exists a unique two-qubit strategy is sufficient for self-testing in this scenario. To do so, note the family of Bell expressions $\bar{B}_{\theta,\phi,\omega}$ contain no marginal terms, so we can apply \cref{lem:reduce1} to obtain a reduced strategy which also maximally violates the Bell expression. If we can show this reduced strategy is unique, then we can apply \cref{lem:reduce2} to show the general two-qubit strategy is unique. 

\begin{lemma}
    Let $\theta,\phi,\omega\in\mathbb{R}$ be such that~\eqref{eq:cond1App} holds. Every two-qubit strategy $(\tilde{\rho}^{\mathsf{q}},\{\tilde{A}_{x}^{\mathsf{q}}\}_x,\{\tilde{B}_{y}^{\mathsf{q}}\}_y)$ that satisfies $\mathrm{Tr}[\tilde{\rho}^{\mathsf{q}}\bar{B}_{\theta,\phi,\omega}] = 0$ is equivalent to the following strategy up to local unitaries: 
    \begin{equation}
    \begin{aligned}
        \tilde{\rho}^{\mathsf{q}} &= \ketbra{\Phi_{0}}{\Phi_{0}}, \\
        \tilde{A}_{0}^{\mathsf{q}} &= \sigma_{Z}, \\
        \tilde{A}_{1}^{\mathsf{q}} &= \cos(\theta) \, \sigma_{Z} + \sin(\theta) \, \sigma_{X} , \\
        \tilde{B}_{0}^{\mathsf{q}} &=  \sin(\phi) \, \sigma_{Z} + \cos(\phi) \, \sigma_{X} , \\
        \tilde{B}_{1}^{\mathsf{q}} &= \sin(\omega) \, \sigma_{Z} + \cos(\omega) \, \sigma_{X} .
    \end{aligned} \label{eq:equiv}
\end{equation}
\label{lem:solve}
\end{lemma}
\begin{proof}
    We start be considering strategies of the form~\eqref{eq:red}, in which we denote the measurement angles by $a_{x}',b_{y}'$, and the eigenvalues of the Bell diagonal state $\lambda_{\alpha}'$. We begin by inserting the SOS decomposition for $\bar{B}_{\theta,\phi,\omega}$ (cf.~\eqref{eq:SOS1}):
    \begin{equation}\label{eq:insertSOS}
        \mathrm{Tr}[\rho\bar{B}_{\theta,\phi,\omega}] = \sum_{\mu=0}^{1}\sum_{\alpha=0}^{3}c_\mu\lambda_{\alpha}'\bra{\Phi_{\alpha}}R_{\mu}^{\dagger}R_{\mu}\ket{\Phi_{\alpha}} = 0,
    \end{equation}
    which implies $\lambda_{\alpha}'\|R_{\mu}\ket{\Phi_{\alpha}}\|^{2} = 0$ for all $\mu,\alpha$. Let us assume $\lambda_{0}'> 0$, so that $\|R_{\mu}\ket{\Phi_{0}}\|^{2} = 0$; the proof follows identically otherwise, up to a local unitary equivalence of the solution. Let $T_{A}$ be the unitary that rotates Alice's first measurement, which is currently defined by an arbitrary angle $a_{0}' = -\zeta$, to $\sigma_{Z}$. That is, $T_{A} = R(\zeta) = \cos(\zeta/2)\mathbb{I}_{2} - \mathrm{i}\sin(\zeta/2)\sigma_{Y}$. Then $T_{B} = R(\zeta)$ satisfies $T\ket{\Phi_{0}} = \ket{\Phi_{0}}$ for all $\zeta$, where $T = T_{A}\otimes T_{B}$. We therefore find $\|R_{\mu}\ket{\Phi_{0}}\|^{2} = 0$ implies
    \begin{equation}
       \|TR_{\mu}T^{\dagger}\ket{\Phi_{0}}\|^{2} = 0.
    \end{equation}
    We label the rotated angles $a_{0}'' = 0$, $a_{1}'' = a_{1}' + \zeta,b_{y}''=b_{y}' + \zeta$. The above constraint with $\mu = 0$ implies
    \begin{equation}
        \begin{aligned}
            \cos(\phi)\sin(a_{1}'') &= \sin(\theta)\sin(b_{0}''),\\
            \cos(\phi)\cos(a_{1}'') &= \sin(\theta)\cos(b_{0}'') + \cos(\theta+\phi).
        \end{aligned} \label{eq:const1}
    \end{equation}
    Summing the squares of both sides, we find $\cos(b_{0}'') = \sin(\phi)$, which implies $\cos(a_{1}'') = \cos(\theta)$. Using $\sin^{2}(a_{1}'') + \cos^2(a_{1}'') = 1$ we find $\sin(a_{1}'') = (-1)^{t}\sin(\theta)$ for some $t \in \{0,1\}$, hence $\sin(b_{0}'') = (-1)^{t}\cos(\phi)$. The constraints for $\mu = 1$ are identical to \cref{eq:const1} with $\phi\mapsto \omega$ and $b_{0}'' \mapsto b_{1}''$, from which we conclude $\sin(b_{1}'') = (-1)^{t}\cos(\omega)$ and $\cos(b_{1}'') = \sin(\omega)$. Note the case $t=1$ is related to the case $t=0$ via the local unitary $\sigma_{Z} \otimes \sigma_{Z}$, which preserves $\ket{\Phi_0}$ and $\ket{\Phi_3}$ and takes $\ket{\Phi_1}$ to $-\ket{\Phi_1}$ and $\ket{\Phi_2}$ to $-\ket{\Phi_2}$.

    Now, for $\alpha \neq 0$, we can explicitly calculate 
    \begin{equation}
    \begin{aligned}
        \lambda_1'\|R_{0}\ket{\Phi_1}\|^{2} &=  \lambda_1'\|TR_{0}T^{\dagger}T\ket{\Phi_1}\|^{2} = 4\lambda_1'\sin^{2}(\theta)\cos^{2}(\phi+(-1)^t\zeta), \\
        \lambda_1'\|R_{1}\ket{\Phi_1}\|^{2} &= \lambda_1'\|TR_{1}T^{\dagger}T\ket{\Phi_1}\|^{2} = 4\lambda_1' \sin^{2}(\theta)\cos^{2}(\omega+(-1)^t\zeta),\\
        \lambda_2'\|R_{0}\ket{\Phi_2}\|^{2} &=  \lambda_2'\|TR_{0}T^{\dagger}T\ket{\Phi_2}\|^{2} = 4\lambda_2'\sin^{2}(\theta)\sin^{2}(\phi+(-1)^t\zeta), \\
        \lambda_2'\|R_{1}\ket{\Phi_2}\|^{2} &= \lambda_2'\|TR_{1}T^{\dagger}T\ket{\Phi_2}\|^{2} = 4\lambda_2' \sin^{2}(\theta)\sin^{2}(\omega+(-1)^t\zeta),\\
        \lambda_3'\|R_{0}\ket{\Phi_3}\|^{2} &=  \lambda_3'\|TR_{0}T^{\dagger}T\ket{\Phi_3}\|^{2} = 4\lambda_3'\sin^{2}(\theta), \\
        \lambda_3'\|R_{1}\ket{\Phi_3}\|^{2} &= \lambda_3'\|TR_{1}T^{\dagger}T\ket{\Phi_3}\|^{2} = 4\lambda_3' \sin^{2}(\theta),
    \end{aligned}
    \end{equation}
    where we inserted the solution for $a_{x}'',b_{y}''$ calculated below \cref{eq:const1}. Due to \cref{eq:insertSOS}, we require each of these expressions to equal $0$. \cref{eq:cond1App} implies $\sin(\theta) \neq 0$, so we must have $\lambda_3'=0$.  Now suppose $\lambda_1'\neq0$. We must have $\phi+(-1)^t\zeta=n\pi/2$ and $\omega+(-1)^t\zeta=m\pi/2$ for odd integers $n,m$, which implies $\phi=\omega+k\pi/2$ for an even integer $k=n-m$. Inserting this into \cref{eq:cond1App} we find the expression on the left hand side is strictly positive, resulting in a contradiction. Therefore we must have $\lambda_1'=0$. For $\lambda_2'\neq0$, a similar contradiction is reached using an analogous argument, allowing us to conclude that $\lambda_\alpha'=0$ for $\alpha\neq0$ and hence that $\lambda_0'=1$.

    There are two strategies (given by $t=0$ and $t=1$), but they are related by $\sigma_Z\otimes\sigma_Z$ which preserves $\ket{\Phi_0}$. Hence, up to local unitaries the solution is unique.

    So far we have assumed a solution of the form~\eqref{eq:red}. However,  Lemma~\ref{lem:reduce2} implies that this assumption can be removed. The form given in~\eqref{eq:equiv} corresponds to one such solution.
\end{proof}
Self-testing is obtained by directly applying \cref{lem:blocks}. The exact state and measurements in \cref{eq:appstrat} can be obtained by rotating into the $X$-$Y$ plane.  
\end{proof}

\section{Proof of Propositions 3--6}\label{app:3to6}
\noindent \textbf{Proposition 3.} For any $s\in(2,5/2]$, and any $\epsilon\in(0,2-(3/4)\log(3)]$, there exists a tuple $(\theta,\phi,\omega)$ satisfying~\eqref{eq:cond1App}, along with a set of quantum correlations achieving $R^{\mathrm{key}}_{\theta,\phi,\omega}(\eta_{\theta,\phi,\omega}^{\mathrm{Q}}) = 1- \epsilon$ and $I_{\mathrm{CHSH}}=s$.

\begin{proof}
    Consider the family of Bell expressions in \ifarxiv\cref{prop:bellExp} \else Proposition~1\fi with $\phi = \pi/2 - \epsilon'$ for some small $\epsilon' > 0$. In this case $H(A|X=0,Y=0,B) = H_{\mathrm{bin}}[(1+\cos \epsilon')/2] := \epsilon$. Now let us restrict $\theta \in (\epsilon',\pi/2)$ and $\omega \in (\pi/2,\pi]$. We first show that~\eqref{eq:cond1App} holds for this range of parameters, which reads
    \begin{equation}
        \sin(\epsilon' - \theta)\sin(\epsilon')\cos(\theta+\omega)\cos(\omega) <0.
    \end{equation}
    We have that $\sin (\epsilon') > 0$, $\sin(\epsilon' - \theta) < 0$, $\cos \omega < 0$, and finally $\cos(\theta + \omega) < 0$, hence satisfying~\eqref{eq:cond1App}. This implies $H(A|X=0,E) = 1$ from the self-test.

    The CHSH value of the underlying strategy in \cref{eq:appstrat} is given by $I_{\mathrm{CHSH}} = \cos (\epsilon') + \sin (\omega) + \cos(\theta - \epsilon') - \sin(\theta+\omega)$. Now we show there exists a tuple $(\theta,\omega,\epsilon')$ in the above range that achieves every $I_{\mathrm{CHSH}} \in (2,5/2)$ when $\epsilon'$ is arbitrarily small. Consider taking $\omega=5\pi/6$, giving $I_{\mathrm{CHSH}}(\theta,\omega = 5\pi/6,\epsilon') = 1/2 + \cos(\epsilon') + \cos(\theta - \epsilon')-\cos(\theta + \pi/3)$.  This is continuous in $\theta$ and $\epsilon'$ when extended to all reals, and for $\epsilon'=0$ it becomes $3/2+\cos(\pi/3-\theta)$, which takes value $2$ for $\theta=0$ and value $5/2$ for $\theta=\pi/3$.  In a bit more detail we have
    \begin{equation}
      I_{\mathrm{CHSH}}(\theta=2\epsilon',\omega = 5\pi/6,\epsilon') = \frac{1}{2}\left(3+\cos(2\epsilon')+\sqrt{3}\sin(2\epsilon')\right)<2+\sqrt{3}\epsilon',
    \end{equation}
    hence we can get arbitrarily close to $I_{\mathrm{CHSH}}=2$ by taking $\epsilon'$ small enough.
    
    In addition,
    \begin{equation}
    I_{\mathrm{CHSH}}(\theta=\pi/3,\omega = 5\pi/6,\epsilon') = \frac{1}{2}\left(2+3\cos(\epsilon')+\sqrt{3}\sin(\epsilon')\right)>\frac{5}{2},
    \end{equation}
    where the inequality holds for $\epsilon'<\pi/3$. Hence we can achieve $I_{\mathrm{CHSH}}=5/2$ for any $\epsilon'\in(0,\pi/3]$.  Note that $\epsilon'=\pi/3$ gives $\epsilon=2-(3/4)\log(3)\approx0.811$.
    
    Therefore by making $\epsilon$, and consequently $\epsilon'$, arbitrarily small we can get arbitrarily close to the local bound $I_{\mathrm{CHSH}} = 2$, or achieve $I_{\mathrm{CHSH}} = 5/2$, while at the same time generating arbitrarily good key. Moreover, by continuity of $I_{\mathrm{CHSH}}(\theta,\omega=5\pi/6,\epsilon')$ over the interval $\theta \in (\epsilon',\pi/3]$, the entire range $(2,5/2]$ can be accessed when $\epsilon'$ is small enough, thus proving the claim. 
\end{proof}

\noindent \textbf{Proposition 4.} For any $s\in(2,5/2]$, there exists a tuple $(\theta,\phi,\omega)$, along with a set of quantum correlations achieving $r^{\mathrm{key}}_{\theta,\phi,\omega} = 1$ and $I_{\mathrm{CHSH}}=s$.

\begin{proof}
We establish this proof using the results of Ref.~\cite{Wang_2016}, who found all correlations that self-test the singlet in this scenario. This set of correlations coincides with the nonlocal boundary of $\mathcal{Q}_{\mathrm{corr}}$. We state this theorem below:
\begin{theorem}[\cite{Wang_2016}, Theorem 1]
    The observed correlations $\langle A_{x}B_{y}\rangle $ self-test the singlet if and only if they satisfy one of the eight conditions:
\begin{equation}
    \sum_{(x,y) \neq (i,j)}\arcsin{\langle A_{x}B_{y} \rangle } - \arcsin{\langle A_{i}B_{j} \rangle} = \xi \pi,
\end{equation}
where $i,j\in\{0,1\}$ and $\xi \in \{+1,-1\}$, and where there is at most one pair of values of $x,y$ for which $\arccos{\langle A_{x}B_{y} \rangle}\in\{0,\pi\}$. \label{thm:wang}
\end{theorem}
\noindent Proof of the proposition then relies on checking the correlations of the strategy in \cref{eq:appstrat} satisfy the above criteria, whist simultaneously achieving the range of CHSH scores $(2,5/2]$. Set $\phi = \pi/2$, $\omega = 5\pi/6$ and let $\theta \in (0,\pi/3]$. Then the CHSH value is given by $I_{\mathrm{CHSH}} = 3/2 + \cos(\theta) - \cos(\theta + \pi/3) \in (2,5/2]$. Now the correlators are given by
\begin{equation}
    \langle A_{0} B_{0} \rangle = 1, \ \langle A_{0}B_{1} \rangle = 1/2, \ \langle A_{1}B_{0} \rangle = \cos(\theta), \ \langle A_{1}B_{1} \rangle = \sin(\pi/6 - \theta),
\end{equation}
and their $\arccos$ values are $0,\pi/3,\theta,\theta + \pi/3$ respectively; for $\theta \in (0,\pi/3]$ the $\arccos$ condition of \cref{thm:wang} is satisfied. Consider the $\arcsin$ condition for $i=1,j=1$. This reads:
\begin{equation}
    \arcsin(1) + \arcsin(1/2) + \arcsin(\cos(\theta)) - \arcsin(\sin(\pi/6-\theta)) = \pi/2 + \pi/6 + \pi/2 - \theta - \pi/6 + \theta = \pi,
\end{equation}
for $\theta \in (0,\pi/3]$. Hence \cref{thm:wang} is satisfied, and the correlations $\langle A_{x}B_{y} \rangle$ self-test the singlet. As a result, the global state shared by Alice and Bob cannot be entangled with Eve, and the entropy $H(A|E,X=0)$ is that of their observed statistics, i.e., $H(A|E,X=0) = 1$. Moreover, the reconciliation term when Bob measures $Y=0$ is equal to $H(A|X=0,Y=0,B) = 0$, resulting in a final rate of $r^{\mathrm{key}}_{\theta,\phi,\omega} = 1$ as claimed, for $\phi = \pi/2,\omega = 5\pi/6$ and every $\theta \in (0,\pi/3]$, which covers the range of CHSH values $(2,5/2]$.
\end{proof}

\noindent \textbf{Proposition 5.} For any $s\in(2,1+\sqrt{2}]$ and any $\epsilon\in\left(0,H_{\mathrm{bin}}((2+\sqrt{2})/4)\right)$, there exists a tuple $(\theta,\phi,\omega)$ satisfying~\eqref{eq:cond1App}, along with a set of quantum correlations achieving $R^{\mathrm{global}}_{\theta,\phi,\omega}(\eta_{\theta,\phi,\omega}^{\mathrm{Q}})=2$, $R^{\mathrm{key}}_{\theta,\phi,\omega}(\eta_{\theta,\phi,\omega}^{\mathrm{Q}})=1-\epsilon$ and $I_{\mathrm{CHSH}}=s$.
\begin{proof}
  This claim is obtained in the same way as \ifarxiv\cref{prop:key}\else Proposition~3\fi. As before, we let $\phi = \pi/2 - \epsilon'$ for some small $\epsilon' > 0$, and observe $H(A|X=0,Y=0,B) = \epsilon$. We consider the same range $\theta \in (\epsilon',\pi/2)$ and $\omega \in (\pi/2,\pi]$, for which~\eqref{eq:cond1App} holds, implying $H(A|X=0,E) = 1$ from the self-test. Next we further set $\omega = \pi$, and observe that the self-test also implies $H(AB|X=0,Y=1,E)=2$.

  Finally, note that for this choice we have $I_{\mathrm{CHSH}}(\theta,\omega=\pi,\epsilon')=\cos(\epsilon')+\cos(\theta-\epsilon')+\sin(\theta)$, which extending to all reals and taking $\epsilon'=0$, equals $2$ for $\theta=0$ and $1+\sqrt{2}$ for $\theta=\pi/4$. More precisely,
\begin{equation}
  I_{\mathrm{CHSH}}(\theta = 2\epsilon',\omega=\pi,\epsilon')= 2\cos(\epsilon')+\sin(2\epsilon')<2+2\epsilon',
\end{equation}
so we can get arbitrarily close to $I_{\mathrm{CHSH}}=2$ by taking $\epsilon'$ small enough.

In addition,
    \begin{equation}
I_{\mathrm{CHSH}}(\theta=\pi/4,\omega=\pi,\epsilon')=\frac{1}{\sqrt{2}}+\cos(\epsilon')+\sin(\pi/4+\epsilon')> 1 + \sqrt{2},
    \end{equation}
    where the last line holds for $\epsilon'<\pi/4$. Note that $\epsilon'=\pi/4$ gives $\epsilon=H_{\mathrm{bin}}((2+\sqrt{2})/4)\approx0.601$.
    
    Therefore, using continuity, by taking $\epsilon$, and hence $\epsilon'$ small enough, we can achieve any $I_{\mathrm{CHSH}} \in (2,1+\sqrt{2}]$ with $R_{\theta,\phi,\omega}^{\mathrm{key}}=1-\epsilon$ and $R_{\theta,\phi,\omega}^{\mathrm{global}}=2$.
\end{proof}

\noindent \textbf{Proposition 6.} For any $s \in (2,1+\sqrt{2}]$, there exists a tuple $(\theta,\phi,\omega)$, along with a set of quantum correlations achieving $r^{\mathrm{global}}_{\theta,\phi,\omega} = 2$, $r^{\mathrm{key}}_{\theta,\phi,\omega} = 1$, and $I_{\mathrm{CHSH}} = s$.

\begin{proof}
    We follow the approach of \ifarxiv\cref{prop:key_corr}\else Proposition~4\fi. By setting $\phi = \pi/2, \omega = \pi$ and $\theta \in (0,\pi/4]$, we find the CHSH value covers the range $(2,1+\sqrt{2}]$. The correlators are given by
    \begin{equation}
    \langle A_{0} B_{0} \rangle = 1, \ \langle A_{0}B_{1} \rangle = 0, \ \langle A_{1}B_{0} \rangle = \cos(\theta), \ \langle A_{1}B_{1} \rangle = -\sin(\theta),
\end{equation}
and the $\arccos$ values are $0,\pi/2,\theta,\pi/2 + \theta$, so the $\arccos$ condition of \cref{thm:wang} is satisfied. We then find the $\arcsin$ condition with $i=1,j=1$ is satisfied, completing the proof using the same reasoning as above.  
\end{proof}

\section{Equivalence between the Bell expression in \ifarxiv\cref{prop:bellExp} \else Proposition~1 \fi and that of Le \textit{et al}.}\label{app:C}

Here we show equivalence between our presentation and the original presented in Proposition~14 of Ref~\cite{Le2023quantumcorrelations}. We state this proposition below.
\begin{proposition}[\cite{Le2023quantumcorrelations}, Proposition 14]
    \textit{Let $\alpha,\beta,\gamma,\delta$ satisfy $\alpha+\beta+\gamma+\delta=0$ and $\Delta = \sin\alpha \sin\beta \sin\gamma \sin \delta < 0$ and let $\bm{c} = [\cos\alpha,\cos\beta,\cos\gamma,\cos\delta]^{\mathrm{T}}$. Define
\begin{equation}
    \bm{f} = \frac{1}{K}\Big[ \frac{1}{\sin \alpha},\frac{1}{\sin \beta} , \frac{1}{\sin \gamma} , \frac{1}{\sin \delta}\Big]^{\mathrm{T}},
\end{equation}
where $K = \cot \alpha + \cot \beta + \cot \gamma + \cot \delta$. Then $\bm{f} \bm{\cdot} \bm{c}' \leq 1$ for all $\bm{c}' \in \mathcal{Q}_{\mathrm{corr}}$, with equality if and only if $\bm{c}' = \bm{c}$.}
\end{proposition}
\noindent Using the relation $\delta = -(\alpha+\beta+\gamma)$ we can rewrite the entries in $\bm{f}$:
\begin{multline}
    \bm{f} = \Bigg[ \frac{\sin\beta \sin \gamma \sin(\alpha+\beta+\gamma)}{\sin(\gamma + \alpha)\sin(\gamma+\beta)\sin(\alpha+\beta)},\frac{\sin\alpha \sin \gamma \sin(\alpha+\beta+\gamma)}{\sin(\gamma + \alpha)\sin(\gamma+\beta)\sin(\alpha+\beta)} , \\ \frac{\sin\alpha \sin \beta \sin(\alpha+\beta+\gamma)}{\sin(\gamma + \alpha)\sin(\gamma+\beta)\sin(\alpha+\beta)} ,\frac{-\sin\alpha \sin \beta \sin\gamma}{\sin(\gamma + \alpha)\sin(\gamma+\beta)\sin(\alpha+\beta)} \Bigg]^{\mathrm{T}}.
\end{multline}
Our expression is given by $\bm{B}_{\theta,\phi,\omega}$, and, once normalized by the maximum quantum value $\eta^{\mathrm{Q}}_{\theta,\phi,\omega}$ so the RHS is equal to one, it reads
\begin{multline}
    \bm{B}_{\theta,\phi,\omega} / \eta^{\mathrm{Q}}_{\theta,\phi,\omega} = \Bigg[ \frac{\cos(\theta+\phi)\cos(\theta+\omega)\cos\omega}{\sin\theta \sin(\omega-\phi)\sin(\theta+\omega+\phi)} , \frac{-\cos(\theta+\phi)\cos(\theta+\omega)\cos \phi}{\sin \theta \sin(\omega-\phi)\sin(\theta+\omega+\phi)}, \\ \frac{-\cos \phi \cos \omega \cos(\theta+\omega)}{\sin \theta \sin(\omega-\phi)\sin(\theta+\omega+\phi)} , \frac{\cos \phi \cos \omega \cos(\theta+\omega)}{\sin \theta \sin(\omega-\phi)\sin(\theta+\omega+\phi)} \Bigg]^{\mathrm{T}}.
\end{multline}
By making the substitution $\theta = -\alpha-\gamma$, $\phi = \pi/2 + \alpha$, $\omega = \pi/2 - \beta$, we find $\bm{B}_{\theta,\phi,\omega} / \eta^{\mathrm{Q}}_{\theta,\phi,\omega} = \bm{f}$ and that $\Delta<0$ is equivalent to \cref{eq:cond1App}.

\section{Recovering known Bell inequalities}\label{app:knownB}
In this appendix, we show how many known Bell inequalities from the literature can be recovered from the Bell expression first introduced in~\cite{Le2023quantumcorrelations}. Consider evaluating the Bell expression in \ifarxiv\cref{prop:bellExp} \else Proposition~1\fi for $\theta = -\pi/2,\phi=3\pi/4,\omega=\pi/4$. Then the expression reads
\begin{equation}
    \frac{1}{2\sqrt{2}}\big( \langle A_{0}(B_{0}+B_{1}) \rangle + \langle A_{1}(B_{0}-B_{1})\rangle \big),
\end{equation}
with $\eta^{\mathrm{Q}}_{\theta,\phi,\omega} = 1$, equivalent to the CHSH expression. 

By setting $\phi = 0$, $\omega = \theta = \delta + \pi/2$ for $\delta \in (0,\pi/6]$, and dividing by $-\sin^{2}(\delta)\cos(2\delta)$, we recover the $I_{\delta}$ family of expressions from~\cite{WBC}:
\begin{equation}
    \langle A_{0}B_{0} \rangle + \frac{1}{\sin \delta} \langle A_{0}B_{1} + A_{1}B_{0} \rangle - \frac{1}{\cos(2\delta)} \langle A_{1}B_{1} \rangle,
\end{equation}
with $\eta^{\mathrm{Q}}_{\theta,\phi,\omega} = 2\cos^{3}\delta/(\sin(\delta)\cos(2\delta))$. We can also recover the $J_{\gamma}$ family of expressions from~\cite{WBC}. By setting $\theta = -2\phi/3$ and $\omega=\phi/3$, we find the expression
\begin{equation}
    \cos^3(\phi / 3)\langle A_{0}B_{0} \rangle - \cos(\phi)\cos^2 (\phi/3)\big( \langle A_{1}B_{0} \rangle + \langle A_{0}B_{1} \rangle - \langle A_{1}B_{1} \rangle \big),
\end{equation}
with a quantum bound $\eta^{\mathrm{Q}}_{\theta,\phi,\omega}= \sin^3 (2\phi/3)$. By dividing through by $\cos^3(\phi / 3)$, and setting $\phi = 3\pi/2 - 3\arcsin\big[ \sin( \pi/6 + \gamma )\big]$ we recover the exact expression and quantum bound in~\cite{WBC}. 

Next we consider the family of tilted CHSH expressions without the marginal term from~\cite{AcinRandomnessNonlocality}, called $I_\alpha^0$ there. By choosing $\theta = \pi/2$ and $\omega = -\phi$, and dividing through by $-\sin(\phi)\cos(\phi)$ we find
\begin{equation}
    \sin(\phi)\langle A_{0}(B_{0} - B_{1}) \rangle + \cos(\phi)\langle A_{1}(B_{0} + B_{1} ) \rangle,
\end{equation}
with a quantum bound equal to $2$, which, after relabeling Alice's inputs and dividing through by $\sin \phi$ is equivalent to the expression in~\cite{AcinRandomnessNonlocality}, with a tilting parameter $\alpha = 1/\tan \phi$. 

Finally, we consider the two parameter family of Bell expressions which are symmetric under party permutation~\cite{WBC2}, that is both Alice and Bob have the same angle between their measurements. By setting $\theta = \omega - \phi$, and changing variables $2\omega - \phi = 2\alpha$, $\phi = -2\beta$, we find
\begin{equation}
    \cos^{2}(\alpha - \beta)\cos 2\alpha \langle A_{0}B_{0}\rangle - \cos(\alpha - \beta) \cos 2 \alpha \cos 2 \beta ( \langle A_{0}B_{1} \rangle + \langle A_{1}B_{0} \rangle ) + \cos^{2}(\alpha - \beta)\cos 2\beta \langle A_{1}B_{1}\rangle, 
\end{equation}
with a quantum bound of $2\sin^{2}(\alpha + \beta)\sin(\alpha-\beta)\cos(\alpha-\beta)$, which is equivalent to the expression $J_{\beta,\alpha}$ in~\cite{WBC2} after dividing through by $\cos(\alpha-\beta)$. 

\section{Additional Property of the family of Bell expressions in \ifarxiv\cref{prop:bellExp}\else Proposition~1\fi}\label{app:E}
In this section, we show that the expressions in \ifarxiv\cref{prop:bellExp} \else Proposition~1 \fi cover the boundary of the quantum set of correlators, $\mathcal{Q}_{\mathrm{corr}}$. This has already been shown by Le \textit{et al.}~\cite{Le2023quantumcorrelations}, and Barizien \textit{et al.}~\cite{barizien2023}. For completeness we present a self-contained proof in this appendix, which builds on the criteria for membership of $\mathcal{Q}_{\mathrm{corr}}$ presented by Masanes~\cite{Masanes1,Masanes2}, and the subsequent parameterization by Wang \textit{et al.}~\cite{Wang_2016}. 
\begin{theorem}[\cite{Masanes1}] The correlations $\langle A_{x}B_{y} \rangle$ belong to $\mathcal{Q}_{\mathrm{corr}}$ if and only if, 
\begin{equation}
    -\pi \leq \sum_{(x,y) \neq (i,j)}\arcsin{\langle A_{x}B_{y} \rangle} - \arcsin{ \langle A_{i}B_{j} \rangle} \leq \pi, \ \forall i,j \in \{0,1\}.  \label{eq:mas}
\end{equation}
\end{theorem}

From this we have the following lemma.
    \begin{lemma}
        For every correlation $\bm{c} = [\langle A_{0}B_{0}\rangle,\langle A_{0}B_{1}\rangle,\langle A_{1}B_{0}\rangle,\langle A_{1}B_{1}\rangle]^{\mathrm{T}}$ belonging to the boundary, $\partial \mathcal{Q}_{\mathrm{corr}}$, of $\mathcal{Q}_{\mathrm{corr}}$, there exists a $(\theta,\phi,\omega) \in \mathbb{R}^{3}$ such that $\langle B_{\theta,\phi,\omega} \rangle = \eta^{\mathrm{Q}}_{\theta,\phi,\omega}$.
        \label{lem:bound}
    \end{lemma}
        \begin{proof}
        For this, we follow the parameterization of~\cite{Wang_2016}. Let each of the correlators take the form $\langle A_{x}B_{y}\rangle=\cos(\alpha_{xy})$ with $\alpha_{xy}\in[0,\pi]$, and suppose $\{\langle A_{x}B_{y} \rangle\}_{x,y}$ lie on the boundary of $\mathcal{Q}_{\mathrm{corr}}$. Then one of the conditions in \cref{eq:mas} must be satisfied with equality. Since all the 8 inequalities are equivalent up to relabelling, we can choose $i=0,j=1$ and the upper bound without loss of generality, which implies $\arccos[\langle A_{0}B_{0} \rangle ] + \arccos[\langle A_{1}B_{0} \rangle ] = \arccos[\langle A_{0}B_{1} \rangle ] - \arccos[\langle A_{1}B_{1} \rangle ]$, hence
    \begin{equation}
        \alpha_{00} + \alpha_{10} = \alpha_{01} - \alpha_{11}. 
    \end{equation}
    Let us define $\vartheta = \alpha_{00} + \alpha_{10} = \alpha_{01} - \alpha_{11}$. The vector of correlators is then given by $\bm{c} = [ \cos \alpha_{00} , \cos \alpha_{01} , \cos(\vartheta - \alpha_{00}) , \cos(\vartheta - \alpha_{01})]^{\mathrm{T}}$. By choosing $\theta = \vartheta$, $\phi = \pi/2 - \alpha_{00}$ and $\omega = \pi/2 - \alpha_{01}$, one can verify
    \begin{equation}
        \bm{B}_{\theta = \vartheta,\phi=\pi/2-\alpha_{00},\omega=\pi/2-\alpha_{01}} \bm{\cdot} \bm{c} = \eta_{{\theta = \vartheta,\phi=\pi/2-\alpha_{00},\omega=\pi/2-\alpha_{01}}}^{\mathrm{Q}},
    \end{equation}
    where $\bm{B}_{\theta,\phi,\omega} = [\cos(\theta+\phi)\cos(\theta+\omega)\cos \omega , -\cos(\theta+\phi)\cos(\theta+\omega)\cos \phi , - \cos \phi \cos \omega \cos(\theta+\omega) , \cos \phi \cos \omega \cos(\theta+\phi)]^{\mathrm{T}} $ is a vector encoding the Bell inequality. Hence we have shown, given any correlators that lie on the boundary of $\mathcal{Q}_{\mathrm{corr}}$, we can find a $(\theta,\phi,\omega)$ such that $\langle B_{\theta,\phi,\omega} \rangle = \eta_{\theta,\phi,\omega}^{\mathrm{Q}}$ for those correlations.
\end{proof}

\end{document}